\documentclass[12pt,draftclsnofoot,onecolumn]{IEEEtran}
\usepackage[T1]{fontenc}
\usepackage{cite}
\usepackage{lipsum}
\usepackage{amsthm,amsmath,amssymb}	 
\usepackage{graphicx}
\usepackage{epstopdf}
\usepackage{algpseudocode,algorithm,algorithmicx}
\usepackage{subcaption}
\usepackage[font={small}]{caption}
\usepackage{breqn} 
\usepackage{ctable} 
\renewcommand{\vec}[1]{\mathbf{#1}}	

\newtheorem{theorem}{Theorem} 	
\newtheorem{lemma}{Lemma}

\usepackage{multirow}
\usepackage{enumitem}

\newcommand{\ignore}[1]{}
\algrenewcommand\algorithmicwhile{\textbf{nodes}}
\algrenewtext{EndWhile}{\algorithmicwhile\ \algorithmicend}

\algrenewcommand\algorithmicrequire{\textbf{Precondition:}}
\algrenewcommand\algorithmicensure{\textbf{Postcondition:}}
\def\NoNumber#1{{\def\alglinenumber##1{}\State #1}\addtocounter{ALG@line}{-1}}

\ifCLASSINFOpdf
\else
\fi
\begin{document}
%
\title{Position-Constrained Stochastic Inference for Cooperative Indoor Localization}
%
%
%

\author{Rico~Mendrzik,~\IEEEmembership{Student Member,~IEEE,}
        Gerhard~Bauch,~\IEEEmembership{Fellow,~IEEE}%
\thanks{R. Mendrzik and  G. Bauch are with the Institute of Communications, Hamburg University of Technology, Hamburg,
 21073 Germany e-mail.}}

%
%

\markboth{}%
{Mendrzik \MakeLowercase{\textit{et al.}}:Position-Constrained Stochastic Inference for Cooperative Indoor Localization}
%



\maketitle

\begin{abstract}
We address the problem of distributed cooperative localization in wireless networks, i.e. nodes without prior position knowledge (agents) wish to determine their own positions. In non-cooperative approaches, positioning is only based on information from reference nodes with known positions (anchors). However, in cooperative positioning, information from other agents is considered as well. Cooperative positioning requires encoding of the uncertainty of agents' positions. To cope with that demand, we employ stochastic inference for localization which inherently considers the position uncertainty of agents. However, stochastic inference comes at the expense of high costs in terms of computation and information exchange. To relax the requirements of inference algorithms, we propose the framework of position-constrained stochastic inference, in which we first confine the positions of nodes to feasible sets. We use convex polygons to impose constraints on the possible positions of agents. By doing so, we enable inference algorithms to concentrate on important regions of the sample space rather than the entire sample space. We show through simulations that increased localization accuracy, reduced computational complexity, and quicker convergence can be achieved when compared to a state-of-the-art non-constrained inference algorithm.
\end{abstract}


%
\IEEEpeerreviewmaketitle

\section{Introduction}
\label{section_1}
%
%
%
%
\subsection{Motivation and State of the Art}
\IEEEPARstart{C}{ooperative} localization is gaining more and more attention throughout the research community. The striking advantage over non-cooperative localization is that infrastructure requirements can be greatly relaxed \cite{SWW2007,PAT2005,WLW2009}, i.e. agents require fewer anchors to obtain unambiguous position estimates. Consider the example in Fig.\ref{fig:topology}, where two agents want to localize. Both agents have obtained distance estimates with respect to two anchors. With non-cooperative approaches, the agents cannot determine their positions unambiguously. On the other hand, cooperation among the agents could resolve the ambiguity because each agent becomes a virtual anchor (an anchor with some uncertainty) for the other agent. The primary challenge of cooperative localization breaks down to expressing the uncertainty of agents' positions accurately. 
\begin{figure}[t]
	\centering	      
	\includegraphics[width=.5\columnwidth]{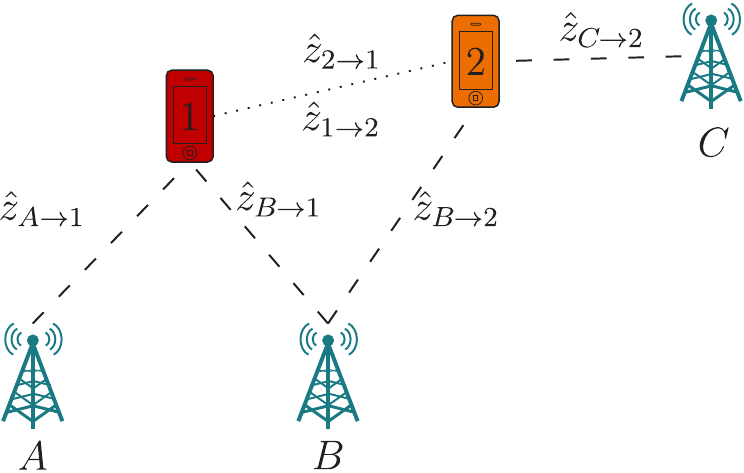}
	\caption{\textit{Running example} - Two agents (orange and red nodes) aim to localize their positions. Each agent has obtained distance measurements (dashed lines) to two anchors (green nodes). Agents cooperate to facilitate unambiguous localization (dotted line). We will use this topology example for the sake of illustration throughout the entire paper.}
	\label{fig:topology}
\end{figure}

Cooperative positioning algorithms can be broadly classified in deterministic and stochastic inference-based approaches. The former class encloses algorithms like (weighted) least squares (LS) \cite{GCW2007,SLK2002,NJSW2015,WCMSGDC2011}, geometric approaches \cite{GWSR2011,GWSR2011b,GWG2013}, multi-hop approaches \cite{NN2001,ND2003,N2004}, and many others. In many cases, deterministic approaches struggle to the account for the uncertainty regarding the positions of agents, and hence they perform poorly compared to stochastic approaches in terms of the positioning accuracy. Stochastic approaches inherently take the position uncertainty of nodes into consideration. Almost all stochastic inference-based approaches can be traced back to some variants of belief propagation (BP). BP is a message passing scheme, in which two main operations are performed, namely, message filtering and message multiplication \cite{YFW2005,LOE2004,KFL2001}. It is known from many inference problems in communication and coding \cite{KFL2001,K2003,MK2002}. In problems from the communication and coding domain, latent random variables are typically discrete. In that case, a belief assigns a certain probability to each state of a random variable (e.g. a bit is 0 with probability $p_0$ and 1 with probability $1-p_0$). In contrast to those problems, the latent random variables in cooperative localization are continuous. Variants of belief propagation that can cope with continuous random variables are based on parametric \cite{LFSWW2012,EMHSW2017,EDHS2013,VWS2012,LWWK2014}, nonparametric \cite{SIFW2003,IFMW2005,SIIFW2010,SZ2013,SZ2009}, and hybrid\footnote{Hybrid approaches employ tools from both the parametric and the nonparametric domain.} \cite{MHW2013,ZCWXZ2017} approaches. 

In parametric belief propagation (PBP), the true beliefs are approximated by parametrized distributions. The parameters of these distributions are determined by minimizing some divergence metric between the true belief and the approximating distribution. The main drawback of PBP is that the parameter determination is only feasible for certain families of distributions \cite{M2001}. Hence complex distributions cannot be represented arbitrarily close. In \cite{LFSWW2012}, the approximating distributions are members of the exponential family and Kullback-Leibler (KL) divergence is minimized to determine the parameters, while in \cite{VWS2012,EMHSW2017,EDHS2013,LWWK2014} the approximating distributions are restricted to the class of Gaussian distributions. In nonparametric belief propagation (NBP), particle representations (sets of samples with associated weights) are used to approximate the true beliefs. If a sufficient number of particles is used, any distribution can be approximated arbitrarily close \cite{KTB2011}. To compute the particle representation of the belief of an agent, importance sampling is employed \cite{SIFW2003,IFMW2005,SIIFW2010,SZ2013,SZ2009}. In importance sampling, the samples are drawn according to a proposal distribution. By adjusting the weights of the samples, the belief is determined. The proposal distributions should resemble the beliefs of agents closely to assure that the samples reside in regions where the beliefs have considerable probability mass. In \cite{SIFW2003,SIIFW2010,IFMW2005}, the proposal distribution is given by any of the messages passed in BP. In contrast to \cite{SIFW2003,SIIFW2010}, so-called parsimonious sampling is proposed in \cite{IFMW2005}, where the proposal distribution generates samples based on the beliefs of all neighbors from the previous iteration.  Other proposal distributions are presented in \cite{SZ2013,SZ2009}. In particular, mixture importance sampling with reference particles \cite{SZ2013} uses a proposal distribution that consists of the sum of all incoming messages from neighbors and sprinkles a certain percentage of the samples uniformly over the considered area to increase the robustness. Auxiliary sampling was also proposed in \cite{SZ2013}, where the proposal distribution is augmented with an auxiliary variable. The auxiliary variable takes the messages from anchors into account leading to samples that are more concentrated in the area close to the true location of a node. Boxed importance sampling is presented in \cite{SZ2009}, where the support of the proposal distribution is constrained heuristically by a rectangle. 
Hybrid belief propagation approaches employ both particle representations and parametrized distributions to represent BP messages and beliefs. Beliefs are represented by particles in \cite{MHW2013}, while other messages are approximated by parametrized distributions. A similar approach is chosen in \cite{ZCWXZ2017}. Messages are represented by members of the exponential family, while beliefs are computed based on particle representations. The proposal distribution, which is used in \cite{ZCWXZ2017}, draws samples based on the results of the previous BP iteration. In general, the proposal distributions which are presented in literature do not resemble the beliefs closely, and particles are utilized inefficiently. 

\begin{figure}[t]
	\centering	      
	\includegraphics[width=.6\columnwidth]{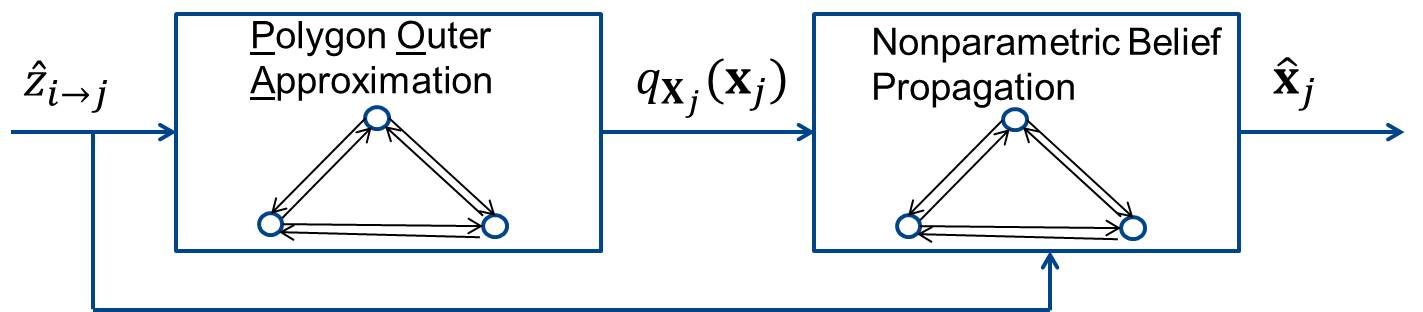}
	\caption{\textit{Schematic of our approach} - Based on the distance estimates $\hat{\mathrm{z}}_{i\rightarrow j}$, we determine our constrained proposal distributions $q_{\vec{X}_j}(\vec{x}_j)$ using POA. Subsequently, we obtain estimates of the positions of the agents $\hat{\vec{x}}_j$ by employing NBP with our novel proposal distribution.}
	\label{fig:two_stage_approach}
\end{figure}
\subsection{Contribution and Paper Organization}
Our essential finding is that constraints on the positions of agents can be leveraged in the inference problem. We propose to  pursue a two-phase approach: first we employ a cheap algorithm called polygon outer-approximation (POA) to determine a set of constraints on the sample space based on the distance estimates $\hat{z}_{i\rightarrow j}$. We use these constraints to determine the proposal distributions $q_{\vec{X}_j}(\vec{x}_j)$. Secondly, we harness these constraints in order to ease the inference problem which we use to obtain the position estimates $\hat{\vec{x}}_j$. Fig. \ref{fig:two_stage_approach} depicts the high-level overview of our approach. Employing our constrained proposal distribution for NBP increases estimation accuracy and relaxes computational requirements. 

The following list contains the main contributions of this paper:
\begin{itemize}
	\item We derive the conditions which have to be met such that the positions of agents can be constrained. In addition, we show that these conditions are fulfilled in the context of ultra-wideband indoor localization.
	\item We develop an algorithm to tightly constrain the positions of agents. 
	\item Based on these constraints, we introduce a novel type of proposal distribution that inherently allows drawing samples from the important regions of the sample space.
	\item Through simulations, we show that increased localization accuracy, quicker convergence, and reduced computation time can be achieved when compared to a state-of-the-art proposal distribution.
\end{itemize}

The remainder of this paper is organized as follows. Section \ref{section_2} introduces our system model and reviews nonparametric belief propagation. In section \ref{section_3}, we derive the conditions necessary to constrain support of the beliefs. We describe our algorithm which outer-approximates the support of beliefs by polygons, and  we show how to exploit the polygon outer-approximations for sampling. Finally, we present our constrained proposal distribution in section \ref{section_4}. Section \ref{section_5} contains the numerical evaluation of our proposal distribution. Section \ref{section_6} concludes the paper.

\section{Fundamentals and Assumptions}
\label{section_2}
In this section, we briefly introduce the system and measurement model. Subsequently, we concisely review the concept of nonparametric belief propagation.
\subsection{System Model}
We consider the problem of cooperative indoor localization using nonparametric belief propagation. 
\subsubsection{Network Topology}
In our setup, agents want to determine their positioning using distance estimates to neighboring agents and anchors. We treat the positions of agents as random variables. The position of agent $j$ is denoted by $\vec{X}_j$. Anchor coverage is assumed to be sparse, i.e. the majority of agents sees only a single anchor. We assume that all nodes are static. Agents can determine range estimates $\hat{z}_{i\rightarrow j}$ to neighboring nodes if neighbors are inside the communication range $r_{\mathrm{com}}$. We denote the set of neighbors of agent $j$ by $\mathcal{S}_{\rightarrow j} = \{ \forall i\neq j ~|~ \left\|\vec{x}_i-\vec{x}_j \right\|_2 \leq r_{\mathrm{com}}\}$. 
\subsubsection{Measurement Model}
\label{subsub:meas_model}
Distance estimates are given by 
\begin{equation}
\hat{z}_{i\rightarrow j} = \left\|\vec{x}_i-\vec{x}_j\right\|_2 + \epsilon_{i\rightarrow j},
\label{eq:dist_est}
\end{equation}
where $ \left\|\vec{x}_i-\vec{x}_j\right\|_2$ is the true distance between node $i$ and node $j$, and $\epsilon_{i\rightarrow j}$ denotes the ranging error. The ranging error generally depends on the receiver characteristics, e.g., thermal noise, distance estimation algorithm, and the environment of transmission. In the indoor environment, multipath and non-line-of-sight (NLOS) propagation are especially critical \cite{LS2002,DCFGW2009,JW2008,FDMW2006}. For instance, NLOS and multipath propagation introduce positive biases in distance estimates. Even in the presence of a line-of-sight (LOS) path, positive ranging errors are mostly observed, because time-of-flight-based distance estimation algorithm set a threshold such that the false alarm (noise peaks causing distance estimates in the absence of a transmission) probability is small. Hence, negative ranging errors are extremely unlikely, if they occur at all \cite{MGWW2010,WMGW2012}. We note that our POA algorithm relies on the assumption of non-negative ranging errors to determine constraints of the sample space. Later in this paper, we show that the true location of a node can only be guaranteed to reside inside the determined polygon if $\epsilon_{i\rightarrow j}\geq 0$. However, in the unlikely case of a negative ranging error, the corresponding distance estimate can be discarded from our proposed POA algorithm, and the resulting polygon is still guaranteed to contain the true location of a node. We employ the model used in \cite{MB2016,VWS2012,GWSR2011b}. The model assumes one-sided non-negative exponentially distributed ranging errors resulting in the likelihood function
\begin{equation}
p(\hat{z}_{i\rightarrow j}|\mathbf{x}_j,\mathbf{x}_i)=
\begin{cases}
\lambda e^{-\lambda\left| \left\|\mathbf{x}_j-\mathbf{x}_i\right\|-\hat{z}_{i\rightarrow j}\right|}, & \hat{z}_{i\rightarrow j}\geq \left\|\mathbf{x}_j-\mathbf{x}_i\right\|\\
0, &\text{else},
\end{cases}
\label{eq:exp_dist}
\end{equation}
where $1/\lambda$ is the mean ranging error. 
\subsection{Nonparametric Belief Propagation for  Cooperative \\Localization}
Let us briefly review the concept of nonparametric belief propagation, with its main two pillars: message filtering and messages multiplication. Note that our aim is not to provide a detailed derivation but to concisely review the concept. An excellent tutorial of the algorithm can be found in \cite{WLW2009}. 

We start with the following factorization of the joint a posteriori distribution of the positions of all agents $\vec{X}$ given all distance estimates $\mathbf{\hat{Z}}$
\begin{equation}
	 p_{\mathbf{X}|\mathbf{\hat{Z}}}(\mathbf{x}|\mathbf{\hat{z}})=\prod_j \prod_{i\in S_{\rightarrow j}} p(\hat{z}_{i \rightarrow j}|\mathbf{x}_i,\mathbf{x}_j)p_{\mathbf{X}_j}(\mathbf{x}_j).
	\label{eq:factorized_APD}
\end{equation} 
\begin{figure}[t]
	\centering	      
	\includegraphics[width=.5\columnwidth]{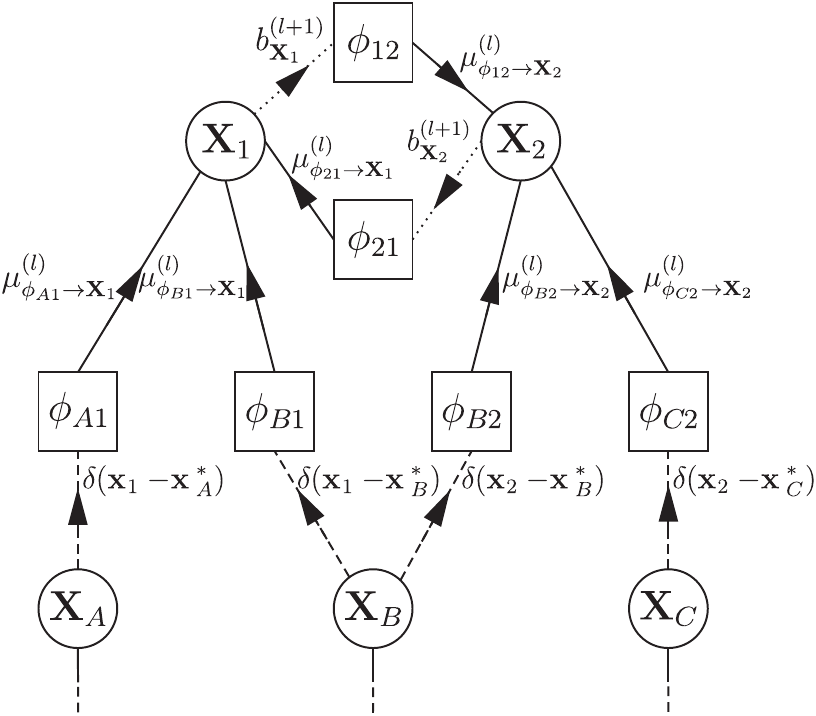}
	\caption{\textit{Factor graph of the running example from Fig. \ref{fig:topology}} - Variable nodes are shown as circular vertices, factor nodes as rectangular vertices, and messages are indicated by arrows. Wireless links are indicated by dotted (inter-agent) and dashed (anchor-agent) lines according to Fig. \ref{fig:topology}. Solid edges determine intra-agent messages.}
	\label{fig:factor_graph}
\end{figure}
NBP employs message passing to determine the marginal distributions $p_{\mathbf{X}_j|\mathbf{\hat{Z}}}(\mathbf{x}_j|\mathbf{\hat{z}}),\forall j$, based on which the position estimates are obtained. In the following, we describe the messages which we encounter in NBP. The factorization in \eqref{eq:factorized_APD} can be visualized by a factor graph, which is a bipartite graph that contains factor nodes, variable nodes, and edges to connect the nodes \cite{W2007}. An exemplary factor graph is depicted in Fig. \ref{fig:factor_graph}. Factor nodes and variable nodes perform a filtering and multiplication operation, respectively. 

Every factor node computes one outgoing message based on its incoming message\footnote{In general, factors nodes can have more than one incoming and outgoing message. For cooperative localization, however, factor nodes have only \textit{one} incoming and outgoing message \cite{WLW2009}.}. Each factor node $\vec{X}_i\in \mathcal{S}_{\rightarrow j}$ computes the outgoing message to variable node $\vec{X}_j$ in the $l^{\text{th}}$ iteration according to the following rule \cite{W2007}
\begin{equation}
	\mu_{\phi_{ij} \rightarrow \mathbf{X}_j}^{(l)}(\mathbf{x}_j)	\propto \int \phi_{ij}(\hat{z}_{i\rightarrow j}|\mathbf{x}_i,\mathbf{x}_j) \cdot \mu_{\mathbf{X}_i \rightarrow \phi_{ij}}^{(l)}(\mathbf{x}_i) \text{d} \mathbf{x}_i,
	\label{eq:message_to_variable_node}
\end{equation}
where $\mu_{\mathbf{X}_i \rightarrow \phi_{ij}}^{(l)}(\mathbf{x}_i)$ denotes the incoming message from $\vec{X}_i$ and $\phi_{ij}(\hat{z}_{i\rightarrow j}|\mathbf{x}_i,\mathbf{x}_j)\triangleq p(\hat{z}_{i \rightarrow j}|\mathbf{x}_i,\mathbf{x}_j)$. Let $\phi_{ij}$ be a shorthand notation for $\phi_{ij}(\hat{z}_{i\rightarrow j}|\mathbf{x}_i,\mathbf{x}_j)$. The operation in (\ref{eq:message_to_variable_node}) is called \textit{message filtering}. Note that the integral in (\ref{eq:message_to_variable_node}) cannot be solved in closed form for nonlinear $\phi_{ij}$ (see e.g. (\ref{eq:exp_dist})) and/or arbitrary $\mu_{\mathbf{X}_i \rightarrow \phi_{ij}}^{(l)}(\mathbf{x}_i)$ \cite{AMGC2002}. Instead, we employ a particle representation to approximate the resulting message, i.e. the continuous message in (\ref{eq:message_to_variable_node}) is approximated by a set of weighted samples 
\begin{equation}
\mathcal{R}{_N}_s\left(\mu_{\phi_{ij} \rightarrow \mathbf{X}_j}^{(l)}(\mathbf{x}_j)\right)=\left\{ w_{ij}^{(k,l)},\mathbf{x}_{ij}^{(k,l)}\right\}_{k=1}^{N_s},
\label{eq:particle_representation}
\end{equation}
where $\mathcal{R}{_N}_s\left(\cdot \right)$ denotes the particle representation of a continuous function using $N_s$ particles. In the context of cooperative localization, message filtering reduces to directly drawing $N_s$ samples with equal weight \cite{LFSWW2012}. Hence message filtering shows linear complexity in the number of samples. Consequently, the complexity is $\mathcal{O}(N_s)$.

Every variable node computes its outgoing message in the $l^{\text{th}}$ iteration as the product of all incoming messages
\begin{equation}
	\mu_{\mathbf{X}_j \rightarrow \phi_{jk}}^{(l+1)}(\mathbf{x}_j)\propto \prod_{i \in S_{\rightarrow j}} \mathcal{R}{_N}_s\left(	\mu_{\phi_{ij} \rightarrow \mathbf{X}_j}^{(l)}(\mathbf{x}_j)\right).
	\label{eq:message_multiplication}
\end{equation}
This operation is called \textit{message multiplication}. The outgoing message in (\ref{eq:message_multiplication}) also constitutes the belief, i.e. $b^{(l+1)}_{\mathbf{X}_j}(\mathbf{x}_j)=\mu_{\mathbf{X}_j \rightarrow \phi_{jk}}^{(l+1)}(\mathbf{x}_j)$. That peculiarity gives rise to the name of the algorithm: belief propagation.  Note that messages ${\mathcal{R}{_N}_s\left(	\mu_{\phi_{ij} \rightarrow \mathbf{X}_j}^{(l)}(\mathbf{x}_j)\right)}$ are given as particle representations. Since the samples $\mathbf{x}_{ij}^{(k,l)}$ are drawn randomly and from independent proposal distributions, they will be distinct with probability one. Direct message multiplication is therefore not possible. To enable multiplication, interpolated versions of these messages are determined using kernel density estimation \cite{IFMW2005}. These densities can, then, be multiplied. Recall that in kernel density estimation, each particle $\left\{ w_{ij}^{(k,l)},\mathbf{x}_{ij}^{(k,l)}\right\}$ \cite{KTB2011}  is coated with a continuous kernel and the superposition of all $N_s$ kernels yields the resulting density
\begin{equation}
\hat{\mu}_{\phi_{ij} \rightarrow \mathbf{X}_j}^{(l)}(\mathbf{x}_j)=\sum_{k=1}^{N_s}{w_{ij}^{(k,l)} K(\mathbf{x}_j;\mathbf{x}_{ij}^{(k,l)},\hat{\boldsymbol\Sigma}_{ij})},
\label{eq:KDE_message}
\end{equation}
where $\hat{\boldsymbol\Sigma}_{ij}$ is estimated using a kernel density estimator. We consider the least squares cross validation estimator from \cite{KTB2011}. Multiplying the kernel density estimates is possible in closed form. However, it requires $\mathcal{O}(N_s^{|\mathcal{S}_{\rightarrow j}|})$ computations, i.e. it scales exponentially in the number of messages. Therefore, we resort importance sampling to approximate the resulting density by a particle representation, i.e. our goal is to obtain a particle representation of the \textit{product of messages} without computing the product explicitly. In importance sampling, we draw $N_s$ samples, $\mathbf{x}_j^{(k,l+1)} ~ k=1,...,N_s$, from a suitable \textit{proposal distribution}, $q_{\mathbf{X}_j}(\mathbf{x}_j)$ \cite{KTB2011}. To obtain a particle representation, we have to assign a proper weight to each sample. The weight accounts for the mismatch of the proposal distribution and the target distribution, which we wish to approximate \cite{KTB2011}. To compute the unnormalized weight of each sample $\tilde{v}_j^{(k,l+1)}$, the quotient of the product of the kernel density estimates of the messages $\prod_{i \in S_{\rightarrow j}} \hat{\mu}_{\phi_{ij} \rightarrow \mathbf{X}_j}^{(l)}(\mathbf{x}_j)$ and the proposal distribution $q_{\mathbf{X}_j}(\mathbf{x}_jin)$ is evaluated for each sample $\mathbf{x}_j^{(k,l+1)}$, i.e.
\begin{equation}
	\tilde{v}_j^{(k,l+1)} \propto  \frac{\prod_{i \in S_{\rightarrow j}} \hat{\mu}_{\phi_{ij} \rightarrow \mathbf{X}_j}^{(l)}(\mathbf{x}_j^{(k,l+1)})}{q_{\mathbf{X}_j}(\mathbf{x}_j^{(k,l+1)})}.
	\label{eq:weight_general}
\end{equation}
The particle representation of the product of messages is then given by the samples $\mathbf{x}_j^{(k,l+1)}$ drawn according to the proposal distribution $q_{\mathbf{X}_j}(\mathbf{x}_j)$ and the normalized weights $v_j^{(k,l+1)}=\tilde{v}_j^{(k,l+1)}/\sum_{k=1}^{N_s}{\tilde{v}_j^{(k,l+1)}}$
\begin{equation}
	\mathcal{R}{_N}_s\left(\mu^{(l+1)}_{\mathbf{x}_j\rightarrow f_k}(\mathbf{x}_j)\right)=\left\{v_j^{(k,l+1)},\mathbf{x}^{(k,l+1)}_j \right\}_{k=1}^{N_s}.
	\label{eq:PR_prod_of_messages}
\end{equation}
\textit{Note that the computationally intensive part is not sampling, but adjusting the weights.} From (\ref{eq:weight_general}), we can infer that adjusting the weights, and thus, message multiplication scales quadratic in the number of samples and linear in the number of incoming messages, $\mathcal{O}\left(|S_{\rightarrow j}|N_s^2\right)$. 

To accurately approximate the product of messages, the proposal distribution $q_{\mathbf{X}_j}(\mathbf{x}_j)$ should generate samples that reside in regions that are close to the true location. Most proposal distributions generate samples by just taking the samples of one of the incoming messages \cite{SIFW2003,SIIFW2010,IFMW2005}. Recall that incoming messages are given as particle representations, and samples are readily obtained. More advanced proposal distributions \cite{IFMW2005}\footnote{Note that in \cite{IFMW2005} two proposal distributions were presented. The more advanced proposal distribution is parsimonious sampling.},\cite{SZ2013,SZ2009}, aim to concentrate the samples in the region of the sample space where the product of messages in \eqref{eq:message_multiplication} has significant probability mass. Our goal is to constrain the support\footnote{The support of a distribution is the part of the sample space with non-zero probability.} of each marginal in order to determine proposal distributions which draw samples only from the relevant regions of the sample space. In other words, if we know that the true marginals have zero probability mass in certain areas of the sample space, it is not worth generating any samples in those areas because their weight will be zero in anyway.

Message multiplication and message filtering are executed iteratively until the beliefs approximate the true marginals closely. We can obtain an estimate on the position of a node in every iteration based on its current belief. Since the belief is given as particle representation $\{v_j^{(k,l+1)},\mathbf{x}_j^{(k,l+1)}\}_{k=1}^{N_s}$, an MMSE estimate of $\mathbf{X}_j$ is obtained by computing the centroid of the particle cloud, i.e. $\hat{\mathbf{x}}_j^{(l+1)}=\sum_{k=1}^{N_s}{v_j^{(k,l+1)} \mathbf{x}_j^{(k,l+1)}}$ \cite{WLW2009}.

\section{Support Outer-Approximation}
\label{section_3}
In this section, we prove the general conditions which need to be met in order to outer-approximate the support of a marginal a posteriori distribution. We conclude this section by showing that these conditions are met in cooperative indoor localization. 

\subsection{General Conditions for Support Outer-Approximation}
Let $p_{\vec{X}_j|\vec{\hat{Z}}}(\vec{x}_j|\vec{\hat{z}})$ be a marginal a posteriori distribution of the factorized joint distribution $p_{\vec{X}|\vec{\hat{Z}}}(\vec{x}|\vec{\hat{z}}) = \prod_{k=1}^K {f_k(\vec{s}_k)}$. The set of random variables which are argument of the $k^{\text{th}}$ factor $f_k$ is denoted by $\vec{s}_k$. Moreover, we define the index set of all factors $f_k$ whose argument contains $\mathbf{x}_j$, i.e. $\mathcal{F}_j = \{k | \mathbf{s}_k \cap \mathbf{x}_j \neq \emptyset \}$.  Our goal is to show that under certain technical conditions
\begin{equation}
	\textsc{supp}\left(p_{\vec{X}_j|\vec{\hat{Z}}}(\vec{x}_j|\vec{\hat{z}})\right)
	\triangleq
	\bigcap_{k\in\mathcal{F}_j } \textsc{supp}\left(f_k \right)
\label{eq:supp_MAP}
\end{equation}
is \textit{compact} and \textit{convex}. Note that the support of the marginals in \eqref{eq:supp_MAP} is generally neither compact nor convex. In the following, we prove the conditions on the factorized joint distribution that need to be fulfilled in order to guarantee that the support of the marginals is compact and convex. In particular, we need compactness to outer-approximate the support, while we need convexity of the support to efficiently the compute outer-approximating polygons. 

\begin{theorem}
The support of the marginal a posteriori distribution $p_{\vec{X}_j|\vec{\hat{Z}}}(\vec{x}_j|\vec{\hat{z}})$ is compact and convex if and only if the supports of the factors $f_k(\vec{s}_k),~ \forall k\in\mathcal{F}_j $ of the joint a posteriori distribution $p_{\vec{X}|\vec{\hat{Z}}}(\vec{x}|\vec{\hat{z}})$ are convex and at least one factor $f_k(\vec{s}_k),~ k\in\mathcal{F}_j$ is compact.
\label{theorem}
\end{theorem}

\begin{proof}
 See Lemma (\ref{lemma:domain_product})-(\ref{lemma:convexity}) in the Appendix.
\end{proof}

\subsection{Applicability of Support Outer-Approximation to Ultra-Wideband Indoor Localization}
Recall that we assume that every agent $j=1,...,N$ has obtained range estimates $\hat{z}_{i\rightarrow j}$ w.r.t. all nodes $i\in S_{\rightarrow j}$. Ranging is performed using time-of-flight estimates with ultra-wideband radios. As discussed in section \ref{subsub:meas_model}, the errors which corrupt range estimates are non-negative \cite{MGWW2010,WMGW2012}, i.e.

\begin{equation}
	\left\|\mathbf{x}_{i}-\mathbf{x}_{j}\right\|\leq\hat{z}_{i\rightarrow j}.
\label{eq:distance_ineq}
\end{equation}
Let us recall the factorization from \eqref{eq:factorized_APD}
\begin{equation}
 p_{\mathbf{X}|\mathbf{\hat{Z}}}(\mathbf{x}|\mathbf{\hat{z}})=\prod_j p_{\mathbf{X}_j}(\mathbf{x}_j)\prod_{i\in S_{\rightarrow j}} p(\hat{z}_{i \rightarrow j}|\mathbf{x}_i,\mathbf{x}_j).
\label{eq:factorized_APD2}
\end{equation}
To prove the applicability of \textbf{Theorem \ref{theorem}} to ultra-wideband indoor localization, we show that every factor in (\ref{eq:factorized_APD2}) is compact and convex. Let us start with the factors of the observation model, $p(\hat{z}_{i\rightarrow j}|\vec{x}_j,\vec{x}_i), \forall i,j$. Regardless of the shape of the density, (\ref{eq:distance_ineq}) ensures that each factor $p(\hat{z}_{i\rightarrow j}|\vec{x}_j,\vec{x}_i)$ of the observation model will be convex and compact, i.e. as long as the distance estimation error is non-negative, the support of the marginal is \textit{guaranteed to be} compact and convex. For instance, let $\vec{x}_i$ be an anchor with true position $\vec{x}_i^{*}$  and $\vec{x}_j$ is an agent. If (\ref{eq:distance_ineq}) is true, the true distance is overestimated and agent $j$ must be inside the disk with radius $\hat{z}_{i\rightarrow j}$ and center $\vec{x}_i^{*}$ with probability one. Now let us consider the prior densities. We typically assume that agents do not have prior knowledge on their position, i.e. the prior distributions $p_{\mathbf{X}_j}(\mathbf{x}_j), \forall j$ are uniform over the plane in which agents want to localize. It can be readily seen that the supports of the priors are convex and compact. Note that other, more informative priors can be assumed, as long as their supports are convex, e.g. a Gaussian prior with arbitrary mean and covariance fulfills these requirements. Thus, the requirements of \textbf{Theorem \ref{theorem}} are met. 

In the following section, we describe how we outer-approximate the support of the marginal a posteriori distributions. Since we are considering two-dimensional positions (localization in a plane), we use polygons which tightly outer-approximate the supports of the marginal a posteriori distributions. Confining the support of the marginal a posteriori distributions of the positions of agents obviously constrains their possible positions. Hence we can employ these polygons to concentrate the attention of the subsequently executed NBP to the relevant region of the sample space.

\textit{Remark:} Our findings from \textbf{Theorem \ref{theorem}} extend beyond the problem of cooperative localization. We emphasize that the support of any marginals can be outer-approximated if the conditions of \textbf{Theorem \ref{theorem}} are met. Only the geometric shape of the outer-approximating objects has to be adjusted to the problem at hand. Thus, it can be verified with \textbf{Theorem \ref{theorem}} whether an inference problem can be treated as a constrained inference problem, which generally simplifies the inference procedure.

\section{Confining the Positions of Nodes}
\label{section_4}
This section contains the main contribution of this paper. We explain how we obtain our novel proposal distributions. First, we outline how we outer-approximate the support of the marginal distributions by polygons. Subsequently, we describe how we embed this side information in the inference problem by incorporating the constraints. Finally, we show how to draw samples from our proposal distributions.
\subsection{Polygon Support Outer-Approximation}
In the previous section, we showed that outer-approximation of the support of the marginal a posteriori distributions is applicable to ultra-wideband indoor localization. We present an algorithm in this section, which confines the location of every agent to a convex polygon. Ideally, for every agent $j$,we wish to find the smallest convex polygon $\mathcal{V}_{j}$ that outer-approximates the support of the marginal distribution $p_{\vec{X}_j|\vec{\hat{Z}}}(\vec{x}_j|\vec{\hat{z}})$, i.e.
\begin{flalign}
 \text{minimize} &\qquad\mathrm{Area}\left({\mathcal{V}_{j}}\right)
 \label{eq:objective1} 											\\  
	\text{subject to} &\qquad 	\textsc{supp}\left(p_{\vec{X}_j|\vec{\hat{Z}}}(\vec{x}_j|\vec{\hat{z}})\right) \subseteq \mathcal{V}_{j}.
	\label{eq:constraint} 																			
\end{flalign}
Since this problem is generally hard to solve, our proposed polygon outer-approximation algorithms attempts to attain the optimal polygon $\mathcal{V}_{j}$ in \eqref{eq:objective1} and \eqref{eq:constraint}. In general, there is no guarantee that the optimal solution is achieved. 

In our previous publication \cite{MB2016}, we used ellipses to outer-approximate the support of marginal a posteriori distributions. In \cite{MB2016}, we adopted an algorithm from Gholami et al. which has its roots in geometrical positioning \cite{GWG2013,GWSR2011,GHO2011}. The algorithm is called \textit{distributed bounding of feasible sets }and it confines the locations of agents to ellipses. To determine the ellipses efficiently, a convex problem formulation is considered. This formulation guarantees convexity but it comes at the expense of unnecessarily loose outer-approximations of the supports. We empirically show in section \ref{section_5} that our polygons-based approach achieves significantly tighter outer-approximations when compared to the elliptical-approach from \cite{GWG2013,GWSR2011,GHO2011}.

Polygon support outer-approximation (POA) iteratively determines a polygon for each agent. Within this polygon, the agent resides with probability one. Two main operations are executed alternately, namely \textit{polygon scaling} and \textit{polygon intersection}. Note that both operations preserve convexity \cite{BV2004}. We emphasize this property when we describe the operations in detail.

In this paragraph, we describe the general procedure of polygon outer-approximation. Subsequently, we describe the operations in detail and use the running example to graphically visualize the operations. 
In all iterations, agent $j$ first receives polygons from the set of neighbors, $\mathcal{V}_{i}^{(l)},~ \forall i \in \mathcal{S}_{\rightarrow j}$. Then, agent $j$ extends these polygons, $\mathcal{V}_{i}^{(l)}$, by the distance measurement $\hat{z}_{i\rightarrow j}$. This operation is called \textit{polygon scaling}. The resulting (scaled) polygons are denoted by $\mathcal{V}_{ij}^{(l)}$. The intersection of the scaled polygons $\mathcal{V}_{ij}^{(l)},~ \forall i \in \mathcal{S}_{\rightarrow j}$ constitutes the polygon of agent $j$ in the next iteration $\mathcal{V}_{j}^{(l+1)}$. This operation is called \textit{polygon intersection}. In contrast to circles and ellipses, the intersection of multiple convex polygons can be determined very efficiently. This procedure is executed in parallel by all agents. After every iteration, agents broadcast their polygons to all neighbors. After a sufficient number of iterations, the size of the polygons converges. Convergence is observed empirically. It is not proven theoretically and requires investigation in future studies. The polygon that each agent has obtained, tightly outer-approximates the support of its marginal a posteriori distribution. The pseudo-code for the algorithm is shown in Algorithm \ref{alg:POA} and a flow diagram is depicted in Fig. \ref{fig:flowchart_POA}.
\begin{figure}[t]
	\centering	      
	\includegraphics[width=.6\columnwidth]{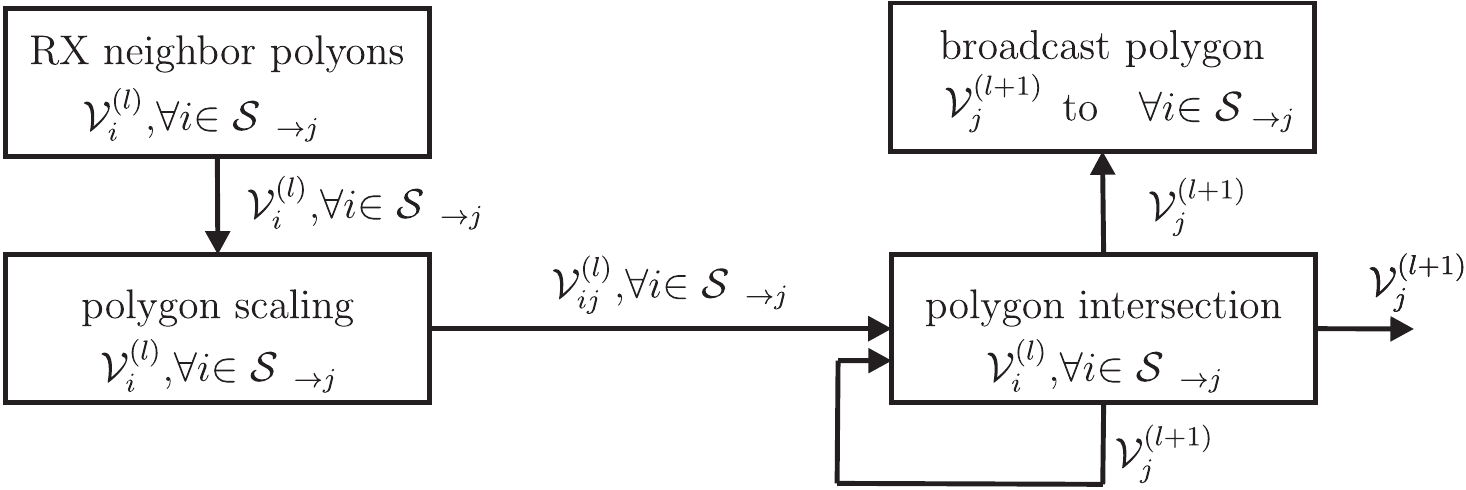}
	\caption{\textit{POA flow diagram}.}
	\label{fig:flowchart_POA}
\end{figure}

Recall that we assume that agents do not have prior knowledge regarding their position. With this assumption, the \textit{first iteration} differs from all subsequent iterations since only the information from anchors is considered. For other priors, information from other agents is also considered. We focus our description on the former case. Let $\mathcal{A}_{\rightarrow j}\subseteq S_{\rightarrow j}$ denote the set of anchors in the communication range of $j$. Then, only  the measurements $\hat{z}_{i\rightarrow j},~\forall i \in \mathcal{A}_{\rightarrow j}$ are considered in the first iteration ($l=1$). Also, recall that anchors have perfect position information. Their position is not confined by a polygon but rather by a single point (their true location), $\vec{x}_j^{*}$. Due to the positive ranging errors, the position of agent $j$ is somewhere inside the disk with radius $\hat{z}_{i\rightarrow j}$ and center $\mathbf{x}^{*}_i$. Agent $j$'s position can be confined to the convex feasible set 
\begin{equation}
	\mathcal{B}_{ij}= \left\{ \mathbf{x}_j \in \mathbb{R}^2 | \left\|\mathbf{x}_j - \mathbf{x}^{*}_i\right\| \leq \hat{z}_{i \rightarrow j} \right\}.
	\label{eq:anchor_set}
\end{equation}

\begin{algorithm}
  \caption{Polygon Support Outer-Approximation
    \label{alg:POA}}
  \begin{algorithmic}[1]
    \Statex
      \State {given $\vec{z}$}  
			\For {$l=1$ to $N_{FS}$}

				\While {$j=1$ to $N$} {in parallel}
					\If {l=1} {(\textit{first iteration})}
						\State receive $\mathbf{x}^{*}_i, \forall i \in \mathcal{A}_j$
						\For {$i \in \mathcal{A}_{\rightarrow j}$}
							\State $\mathcal{V}_{ij}^{(l)}$ = \textit{anchor polygon proc.}($\mathbf{x}^{*}_i, \hat{z}_{i\rightarrow j}$) 
							\NoNumber{- see. \textbf{Algorithm \ref{alg:anchor_polygon}} -}				
						\EndFor		
						\State $\mathcal{V}_j^{(1)}$ = \textit{polygon intersection}($\mathcal{V}_{ij}^{(l)}~ \forall i \in \mathcal{A}_{\rightarrow j}$) 
						\NoNumber{- see. \textbf{Algorithm \ref{alg:polygon_intersection}} -}											
					\Else { (\textit{subsequent iterations})}
								\State broadcast $\mathcal{V}_j^{(l)}$
					\State receive $\mathcal{V}_i^{(l)},~ \forall i \in \mathcal{S}_{\rightarrow j}$
					\For {$i \in \mathcal{S}_{\rightarrow j}$}	
						\State $\mathcal{V}_{ij}^{(l)}$ = \textit{polygon scaling}($\hat{z}_{i\rightarrow j}$, $\mathcal{V}_i^{(l)}$) 
						\NoNumber{- see. \textbf{Algorithm \ref{alg:polygon_scaling}} -}				
					\EndFor
					\State $\mathcal{V}_j^{(l+1)}$ = \textit{polygon intersection}($\mathcal{V}_{ij}^{(l)}~ \forall i \in \mathcal{S}_{\rightarrow j}$) 
					\NoNumber{- see. \textbf{Algorithm \ref{alg:polygon_intersection}} -}				
					\EndIf
				\EndWhile
      \EndFor
  \end{algorithmic}
\end{algorithm}
To obtain a polygon $\mathcal{V}_{ij}^{(1)}$, this disk is outer-approximated by a polygon with $N_E$ edges. This operation is depicted in Fig. \ref{fig:anchor_polygon}. We call such a polygon: \textit{anchor polygon processing}. Mathematically, we describe a convex polygon with $N_E$ edges by an ordered list of $N_E$ vertices or by the intersection of a set of halfspaces. For now, we stick to the description with vertices. 
To generate these vertices, we first determine a polygon outer-approximation of the disk with radius $\hat{z}_{i\rightarrow j}$ which resides in the origin. This can be done efficiently in polar coordinates. Then, we transform the resulting vertices into Cartesian coordinates and shift the vertices of the polygon by the position of the anchor $\mathbf{x}^{*}_i$ to obtain a polygon outer-approximation of $\mathcal{B}_{ij}$ (see Fig. \ref{fig:anchor_polygon}). 

In more detail, we begin with generating points around the origin with uniform angular spacing $\alpha = \frac{2\pi}{N_E}$ and fixed radius 
\begin{equation}
r_{ij} = \frac{\hat{z}_{i\rightarrow j} }{\cos(\alpha/2)}.
\label{eq:POA_radius}
\end{equation} 
We have a list of vertices in polar coordinates, i.e. $\tilde{\vec{v}}_{ij,1}^{(1)} = [\alpha_0,r_{ij}]^T$, $\tilde{\vec{v}}_{ij,2}^{(1)} = [\alpha+\alpha_0,r_{ij}]^T$, ..., $\tilde{\vec{v}}_{ij,N_E}^{(1)} = [(N_E-1)\alpha+\alpha_0,r_{ij}]^T$, where $\alpha_0$ is a random angular offset. Subsequently, we transform these vertices into Cartesian coordinates and shift them by the position of anchor $i$. Let $\bar{\vec{v}}_{ij,m}^{(1)}$ denote the $m^{\text{th}}$ vertex in Cartesian coordinates, the final polygon is given by 
\begin{equation}
	\vec{v}_{ij,m}^{(1)} = \bar{\vec{v}}_{ij,m}^{(1)} + \mathbf{x}^{*}_i, \quad m=1,...,N_E.
	\label{eq:POA_anchor}
\end{equation}
The ordered list of vertices $\vec{v}_{ij,m}^{(1)}$ constitutes the polygon $\mathcal{V}_{ij}^{(1)}=\left\{\vec{v}_{ij,1}^{(1)},\vec{v}_{ij,2}^{(1)},...,\vec{v}_{ij,N_E}^{(1)}\right\}$. Note that the polygon which is defined by $\mathcal{V}_{ij}^{(1)}$ is always convex. 
\begin{figure}[t]
	\centering	      
	\includegraphics[width=.5\columnwidth]{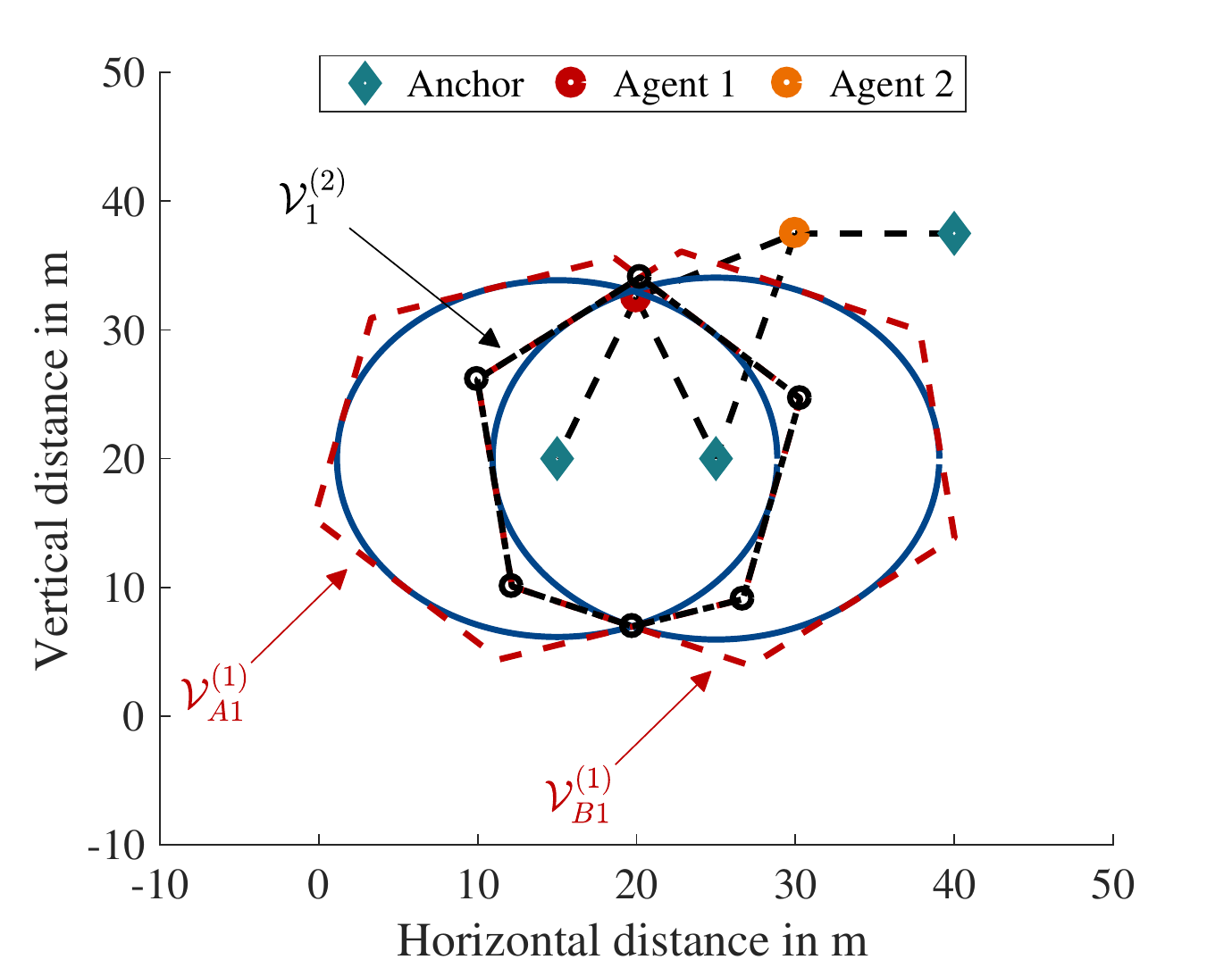}
	\caption{\textit{Anchor polygon} - We outer-approximate the feasible set $\mathcal{B}_{B1}$ by a polygon $\mathcal{V}_{B1}$ with $N_E=6$ edges. The vertices of $\mathcal{V}_{B1}$ are spaced with uniform angle $\alpha = \frac{2 \pi}{N_E} $ around $\mathbf{x}^{*}_B$.}
	\label{fig:anchor_polygon}
\end{figure}
\begin{algorithm}
  \caption{Anchor Polygon Processing
    \label{alg:anchor_polygon}}
  \begin{algorithmic}[1]
    \Statex
      \State {given $\mathbf{x}^{*}_i,\hat{z}_{i\rightarrow j}$}  	
			\State vertex angular spacing $\alpha = \frac{2\pi}{N_E}$
			\State vertex radius $r_{ij} = \frac{\hat{z}_{i\rightarrow j} }{\cos(\alpha/2)}$
			\For {$m=1$ to $N_{E}$}
				\State compute vertex $\vec{v}_{ij,m}^{(1)}=[\alpha_0+m\cdot \alpha,r_{ij}]^T$
				\State $\bar{\vec{v}}_{ij,m}^{(1)} = \textit{pol2cart} (\vec{v}_{ij,m}^{(1)})$
				\State shift vertex $\vec{v}_{ij,m}^{(1)} = \bar{\vec{v}}_{ij,m}^{(1)} + \mathbf{x}^{*}_i$
      \EndFor
			\State determine polygon $\mathcal{V}_{ij}^{(1)}=\left\{\vec{v}_{ij,1}^{(1)},\vec{v}_{ij,2}^{(1)},...,\vec{v}_{ij,N_E}^{(1)}\right\}$
  \end{algorithmic}
\end{algorithm}
When agent $j$ has obtained all anchor polygons $\mathcal{V}_{ij}^{(1)}, ~\forall i \in \mathcal{A}_j$, it intersects these polygons to obtain $\mathcal{V}_{j}^{(1)}$. The intersection of two polygons can be determined efficiently using the Sutherland-Hodgman algorithm \cite{SH1974}. We employ the Sutherland-Hodgman algorithm to intersect pairs of polygons. The Sutherland-Hodgman algorithm selects a convex clipping polygon and a subject polygon; the choice is arbitrary if both polygons are convex. All vertices of the subject polygon are added to an input list. The vertices of this list are updated as the edges of the clipping polygon are considered subsequently. Suppose that the polygon given by $\mathcal{V}_{ij}^{(l)}$ is the clipping polygon, and it is given by a set of vertices ($\vec{v}_{ij,m}^{(l)}, m=1,...,N_E$) or by the intersection of the halfspaces

\begin{equation}
	\left\{ \vec{x}_j | \mathbf{a}_{\mathcal{V}_{ij}^{(l)},m}^T \vec{x}_j \leq c_{ij,m}^{(l)}\right\} \quad m=1,...,N_E,
\label{eq:halfspace_polygon}
\end{equation}
with outward normal vector $\mathbf{a}_{\mathcal{V}_{ij}^{(l)},m} \bot (\vec{v}_{ij,m}^{(l)}-\vec{v}_{ij,m+1}^{(l)})$ and $c_{ij,m}^{(l)}=\mathbf{a}_{\mathcal{V}_{ij}^{(l)},m}^T \vec{v}_{ij,m}^{(l)}$. If we consider the halfspaces instead of the edges, we can update the vertices of the input list. We start with an arbitrary halfspace and all vertices inside that halfspace, i.e. all vertices that reside inside the halfspace are added to an output list. Vertices outside of the halfspace are not added. Subsequently, the subject polygon is traversed. New vertices are added to the output list if the subject polygon intersects with the hyperplane $\left\{ \vec{x}_j | \mathbf{a}_{\mathcal{V}_{ij}^{(l)},i}^T \vec{x}_j = c_{ij,m}^{(l)}\right\}$. After the entire subject polygon has been traversed, the next halfspace of the clipping polygon is considered. The output list of the previous halfspace constitutes the input list to the next halfspace. After every halfspace of the clipping polygon has been considered, the Sutherland-Hodgman algorithm terminates and the resulting output list yields the intersection of the subject and clipping polygon. Note that the number of vertices of the resulting polygon $\mathcal{V}_{j}^{(l+1)}$ can differ from the original number of vertices $N_E$. An example of \textit{polygon intersection} is depicted in Fig. \ref{fig:anchor_polygon_intersection}. The pseudo-code for polygon intersection is given in Algorithm (\ref{alg:polygon_intersection}). 

\textit{Remark}: the polygon intersection operation is identical in every iteration. However, in the first iteration only anchor polygons, $\mathcal{V}_{ij}^{(1)}~ \forall i \in \mathcal{A}_{\rightarrow j}$, are considered, while polygons of anchors and others agents, $\mathcal{V}_{ij}^{(1)}~ \forall i \in \mathcal{S}_{\rightarrow j}$, are considered for subsequent iterations.

\begin{algorithm}
  \caption{Polygon Intersection
    \label{alg:polygon_intersection}}
  \begin{algorithmic}[1]
    \Statex
      \State {given $\mathcal{V}_{ij}^{(l)}~ \forall i \in \mathcal{S}_{\rightarrow j}$}  
			\State select index of the first neighbor $i=\mathcal{S}_{\rightarrow j}(1)$
			\State initialize $\mathcal{V}_{j}^{(l+1)}=\mathcal{V}_{ij}^{(l)}$
			\For {$k=2$ to  $\left|\mathcal{S}_{\rightarrow j}\right|$}
					\State select index of the next neighbor
					\NoNumber{ $i=\mathcal{S}_{\rightarrow j}(k)$}					
					\State intersect two polygons according to \cite{SH1974}
					\NoNumber{$\mathcal{V}_{j}^{(l+1)}=$ \textit{Sutherland-Hodgman}($\mathcal{V}_{j}^{(l+1)},\mathcal{V}_{ij}^{(l)}$)}				
      \EndFor
  \end{algorithmic}
\end{algorithm}

\begin{figure}[t]
	\centering	      
	\includegraphics[width=.5\columnwidth]{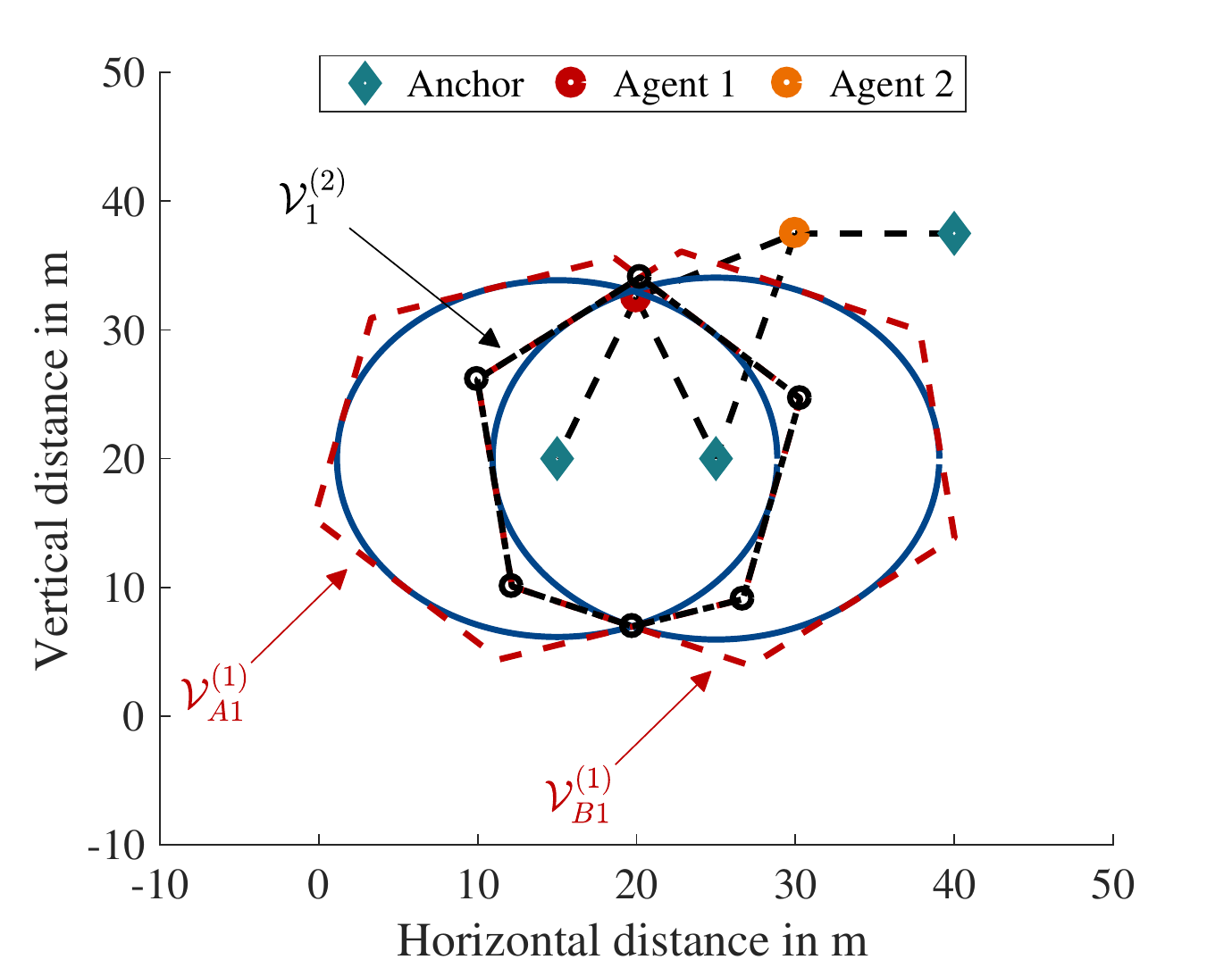}
	\caption{\textit{Polygon intersection }- The intersection of the polygons $\mathcal{V}_{A1}^{(1)}$ and $\mathcal{V}_{B1}^{(1)}$ yields the polygon, $\mathcal{V}_{1}^{(1)}$, which tightly outer-approximates the feasible set (intersection of the two circles). Recall that in the first iteration only anchor polygons are considered. Note that increasing the number of polygon edges, $N_E$, tightens the outer-approximation.}
	\label{fig:anchor_polygon_intersection}
\end{figure}

From the second iteration onward, neighboring agents are also considered. In contrast to an anchor, agent $i$ cannot confine its location to an exact position. However, the position of agent $i$ is confined to the polygon $\mathcal{V}_i^{(l)}, l> 1$. Considering the range estimate to agent $j$, $\hat{z}_{i \rightarrow j}$, the polygon of agent $i$, $\mathcal{V}_i^{(l)}$, has to be extended by the range estimate. To ensure that agent $j$ is inside the polygon $\mathcal{V}_{ij}^{(l)}$, all edges of the polygon $\mathcal{V}_i^{(l)}$ have to be shifted by $\hat{z}_{i \rightarrow j}$ in the direction of the outward pointing normal vector of the respective edge. Recall that convex polygons can be represented by halfspaces according to (\ref{eq:halfspace_polygon}). The scaled polygon is given by the intersection of the halfspaces which are shifted by $\hat{z}_{i \rightarrow j}$ toward the outward normal vector. Wxe obtain these halfspaces by manipulating the right-hand side of all inequalities in (\ref{eq:halfspace_polygon}). Since $c_{i,m}^{(l)}=\mathbf{a}_{\mathcal{V}_{i}^{(l)},m}^T \vec{v}_{i,m}^{(l)}$, where $\vec{v}_{i,m}^{(l)}$ is any point on the corresponding hyperplane, we obtain the shifted halfspace by shifting $\vec{v}_{i,m}^{(l)}$ to 
\begin{equation}
	\vec{\bar{v}}_{i,m}^{(l)}=\hat{z}_{i\rightarrow j} \cdot \vec{a}_{\mathcal{V}_i^{(l)},m} + \vec{v}_{i,m}^{(l)}.
\label{eq:shifted_point}
\end{equation}
The point in (\ref{eq:shifted_point}) yields a point on the $m^{\text{th}}$ shifted hyperlane. Hence the $m^{\text{th}}$ shifted halfspace is given by
\begin{equation}
    	\left\{ \vec{x}_j | \mathbf{a}_{\mathcal{V}_{i}^{(l)},m}^T \vec{x}_j \leq \bar{c}_{i,m}^{(l)}=\mathbf{a}_{\mathcal{V}_{i}^{(l)},m}^T\vec{\bar{v}}_{i,m}^{(l)} \right\}.
\label{eq:shifted_halfspace}
\end{equation}
The scaled polygon is fully described by the set of shifted halfspaces. In order to obtain a list of vertices of the polygon, we have to determine the intersections of adjacent hyperplanes, i.e. for the $m^{\text{th}}$ hyperplane we compute the intersection with the $(m-1)^{\text{th}}$ and $(m+1)^{\text{th}}$ hyperplane. The resulting points constitute the vertices of the scaled polygon $\mathcal{V}_{ij}^{(l)}$. Algorithm \ref{alg:polygon_scaling} shows the pseudo-code for polygon scaling. An illustrative example is shown in Fig. \ref{fig:polygon_intersection_scaled}.

\begin{algorithm}
  \caption{Polygon Scaling
    \label{alg:polygon_scaling}}
  \begin{algorithmic}[1]
    \Statex
      \State {given $\mathcal{V}_{i}^{(l)}, \hat{z}_{i \rightarrow j}$}  
			\For {$m=1$ to  $N_E$}
					\State shift hyperplane according to (\ref{eq:shifted_point}) 
					\NoNumber{$\vec{\bar{v}}_{i,m}^{(l)}=\hat{z}_{i\rightarrow j} \cdot \vec{a}_{\mathcal{V}_i^{(l)},m} + \vec{v}_{i,m}^{(l)}$ }					
					\State determine shifted halfspace according to (\ref{eq:shifted_halfspace})
					\NoNumber{$\left\{ \vec{x}_j | \mathbf{a}_{\mathcal{V}_{i}^{(l)},m}^T \vec{x}_j \leq \bar{c}_{i,m}^{(l)}=\mathbf{a}_{\mathcal{V}_{i}^{(l)},m}^T\vec{\bar{v}}_{i,m}^{(l)} \right\}$}	
      \EndFor
			\State determine intersection of adjacent hyperplanes
			\NoNumber{$\left\{ \vec{x}_j | \mathbf{a}_{\mathcal{V}_{i}^{(l)},m}^T \vec{x}_j = \bar{c}_{i,m}^{(l)} \right\} \cap \left\{ \vec{x}_j | \mathbf{a}_{\mathcal{V}_{i}^{(l)},m+1}^T \vec{x}_j = \bar{c}_{i,m+1}^{(l)} \right\}$}
			\NoNumber{ $\rightarrow \vec{v}_{ij,m}^{(l)},~ \text{for }m=1,...,N_E$}
			\State determine scaled polygon $\mathcal{V}_ij^{(l)}=\left\{\vec{v}_{ij,1}^{(l)},...,\vec{v}_{ij,N_E}^{(l)}]^T\right\}$
  \end{algorithmic}
\end{algorithm}

\begin{figure}[t]
	\centering	      
	\includegraphics[width=.5\columnwidth]{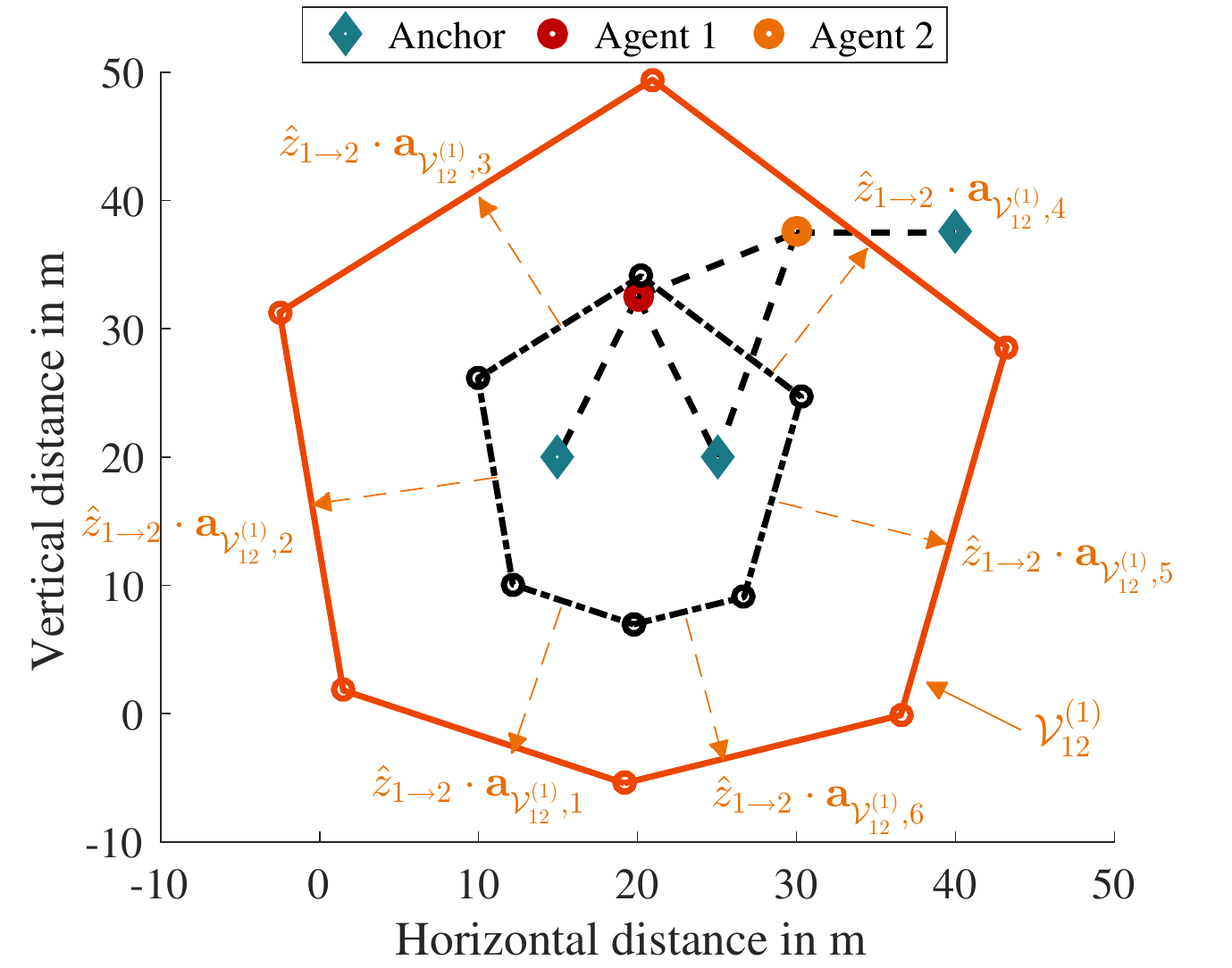}
	\caption{\textit{Polygon scaling }- The resulting polygon of agent $1$ after the first iteration, $\mathcal{V}_{1}^{(1)}$, and its scaled version, $\mathcal{V}_{12}^{(1)}$, are depicted. The original polygon, $\mathcal{V}_{1}^{(1)}$, is given by $N_E=6$ halfspaces, with outward normal vector $\mathbf{a}_{\mathcal{V}_{12}^{(1)},i}, i=1,...,6.$ To scale the original polygon, each halfspace is shifted by $\hat{z}_{12}$ in the direction of its normal vector. The intersections of adjacent halfspaces yield the vertices of the scaled polygon.}
	\label{fig:polygon_intersection_scaled}
\end{figure}

Polygon support outer-approximation (Algorithm \ref{alg:POA}) is readily implemented in a distributed manner, with negligible communication overhead. In each iteration, every agent has to broadcast the vertices of its polygon once. After $N_{FS}$ iterations, every agent has obtained a polygon that outer-approximates the support of its marginal a posteriori distribution. Therefore, also the positions of nodes are confined.

Note that the proposed polygon outer-approximation method does not fail, as long as the assumption in \eqref{eq:distance_ineq} is true, i.e. whenever the range estimates over-estimate the true distance, our Algorithm \ref{alg:POA} produces polygons which are guaranteed to contain the positions of the respective agents. If the assumption in \eqref{eq:distance_ineq} is violated, there are no guarantees that the supports of the marginal distributions are compact and convex. Hence there is no assurance that a polygon exists which contains the true location of the corresponding agent. As discussed in subsection \ref{subsub:meas_model}, the existence of negative range error is extremely unlikely. If a range estimate with negative ranging error occurs, discarding it from Algorithm \ref{alg:POA} will ensure the success of the algorithm. 

\textit{Remark}: Note that the polygon boundaries provide hard-decisions regarding the set of possible location, which might unusual in the context of stochastic inference. However, the one-sided positive measurement errors assign zero-probability to all regions outside of the feasible set. Our polygon outer-approximation approach concentrates all particles inside the feasible set, i.e. the region where the true a posteriori distribution has non-zero probability mass. If any other proposal distribution generates particles outside of the feasible set, the resulting weight of that particle would be computed to be zero, and the particle would not contribute in the estimation process. In other words, the hard-decisions are imposed by the model, not by our polygon-based approach. 

In the following subsection, we describe how we leverage these polygonal constraints to obtain the constrained proposal distributions. 

\subsection{Novel Proposal Distribution}
After $N_{FS}$ iterations of polygon support outer-approximation, every agent $j=1,...,N$ has obtained a polygon $\mathcal{V}_{j}^{(N_{FS})}$ which outer-approximates the support of its marginal a posteriori distribution, $p_{\vec{X}_j|\vec{Z}}(\vec{x}_j|\vec{z})$. Hence we choose these polygons to determine the supports of our proposal distributions $q_{\mathbf{X}_j}(\mathbf{x}_j), \forall j$.  We draw samples $\{\mathbf{x}_j^{(k)} \}_{k=1}^{N_s}$ uniformly over the polygon. Thus, the proposal distribution is given by
\begin{equation}
		q_{\mathbf{X}_j}(\mathbf{x}_j) =
\begin{cases}
\frac{1}{A_{p,j}^{(N_{FS})}}, & \mathbf{x}_j \in \mathcal{V}_{j}^{(N_{FS})} \\
0, &  \text{otherwise},
\end{cases}
	\label{eq:sampling_distribution}
\end{equation} 
where $A_{p,j}^{(N_{FS})}$ is the area of the polygon of the $j^{\text{th}}$ agent.

In order to draw samples uniformly inside a polygon, we use acceptance and rejection sampling. We draw samples uniformly from a rectangle which comprises the polygon $\mathcal{V}_{j}^{(N_{FS})}$. Samples are drawn in horizontal and vertical direction independently, i.e. the edges of the rectangle are aligned with the horizontal and vertical axis. A sample $\mathbf{x}_j^{(k)}$ is accepted if $ \mathbf{x}_j^{(k)} \in \mathcal{V}_{j}^{(N_{FS})}$. The area of the rectangle should be as small as possible in order to achieve the highest acceptance rate. The acceptance rate is given by the ratio of the areas of the polygon and rectangle
\begin{equation}
R_a = \frac{A_{p,j}^{(N_{FS})}}{A_{r,j}^{(N_{FS})}},
\label{eq:acceptance_rate}
\end{equation}
where $A_{r,j}^{(N_{FS})}$ is the area of the rectangle. Consequently, $N_s/R_a$ have to be drawn on average in order to determine obtain $N_s$ accepted samples.

\begin{figure}[t]
	\centering	      
	\includegraphics[width=.5\columnwidth]{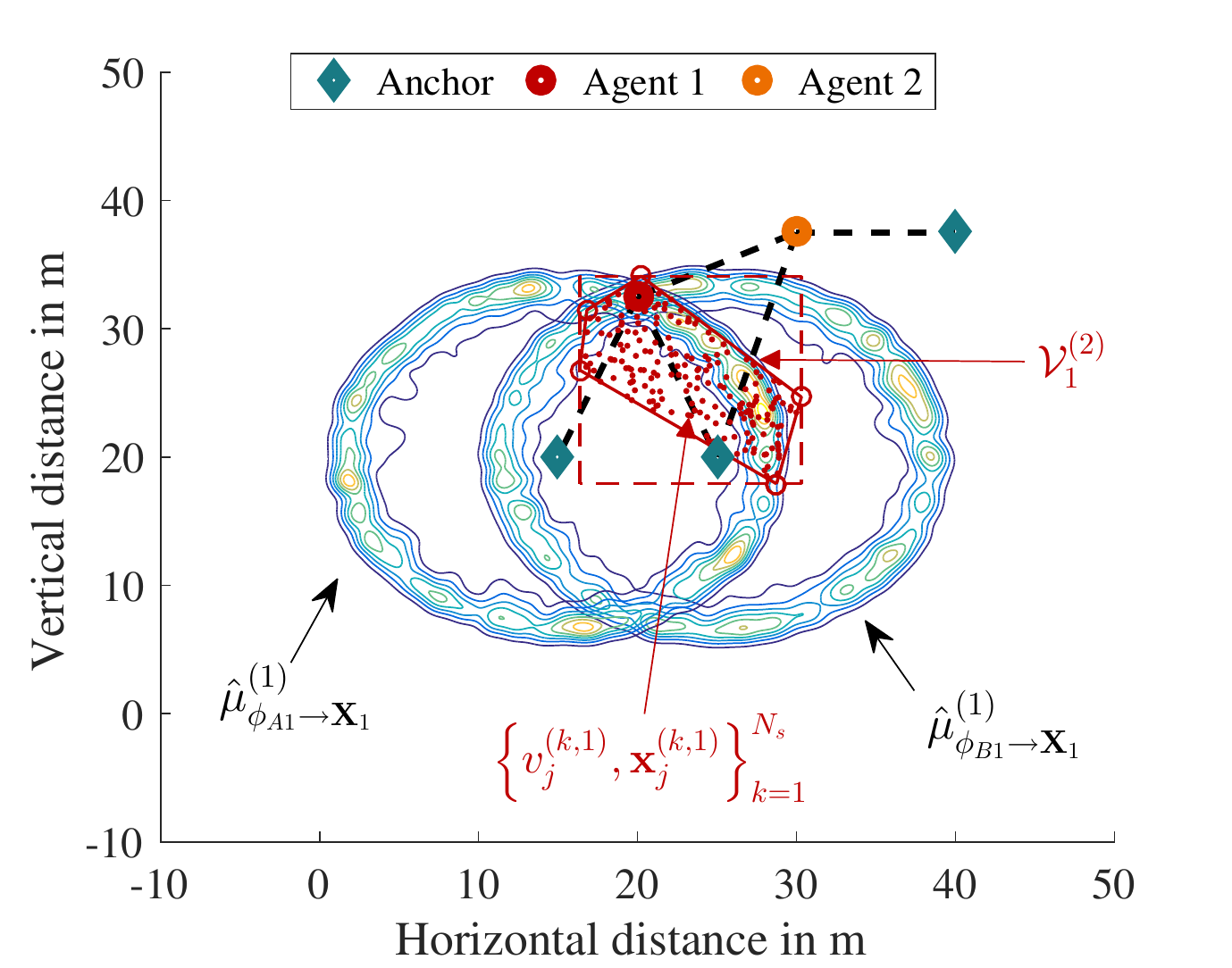}
	\caption{\textit{Novel proposal distribution} - Samples $\{\mathbf{x}_1^{(k)} \}_{k=1}^{N_s}$ drawn according to our novel proposal distribution $q_{\mathbf{X}_1}(\mathbf{x}_1)$. With our proposal distribution, we can confine the region from which we draw samples to the relevant region close to the true location.}
	\label{fig:sampling_distribution}
\end{figure}

Fig. \ref{fig:sampling_distribution} provides an example of samples which are drawn from the our proposal distribution. The figure shows the polygon of agent $1$, $\mathcal{V}_{1}^{(N_{FS})}$, the rectangle for acceptance-and-rejection sampling, the samples, $\left\{ \mathbf{x}_1^{(k)} \right\}_{k=1}^{N_s}$, drawn (red dots), and the kernel density estimates of NBP messages from anchor A and B (contour plots).

\section{Simulation Results}
\label{section_5}
This section contains the numerical evaluation of our novel proposal distribution. In order to determine meaningful parameters for computing our proposal distribution (number of polygon edges $N_E$ and number of iterations $N_{FS}$), we first investigate the polygon support outer-approximation algorithm while disregarding NBP. Following this, we jointly investigate NBP using the polygon-shaped proposal distributions. 
\subsection{Simulation Setup and Performance Measures}
As reference topology, we use a common topology from literature \cite{VWS2012,WLW2009,LFSWW2012}. We simulate a large-scale ultra-wideband network in a 100m x 100m plane, with 100 uniformly distributed agents and 13 fixed anchors. We assume a circular communication range of 20 meters. The ranging errors are distributed according to (\ref{eq:exp_dist}) with a mean of $1/\lambda=0.38m$. Note that in practice the distribution of the ranging errors and its stochastic moments depend on the SNR, the environment, the ranging algorithm itself, and many other aspects.

We quantify localization performance using the average localization error and the outage probability. An agent is said to be in outage if its localization error $e$ exceeds an error threshold $e_{\text{th}}$. We compute the localization error according to $e= \left\|\hat{\mathbf{x}}_j -\mathbf{x}_j \right\|$, where $\hat{\mathbf{x}}_j$ is the estimated location of node $j$, taken as the minimum mean square error (MMSE) estimate of the belief. For statistical significance, we consider 200 random network topologies (positions of agents vary randomly) and collect position estimates at every iteration for every agent. 
\begin{figure}[t]
	\centering	      
	\includegraphics[width=.5\columnwidth]{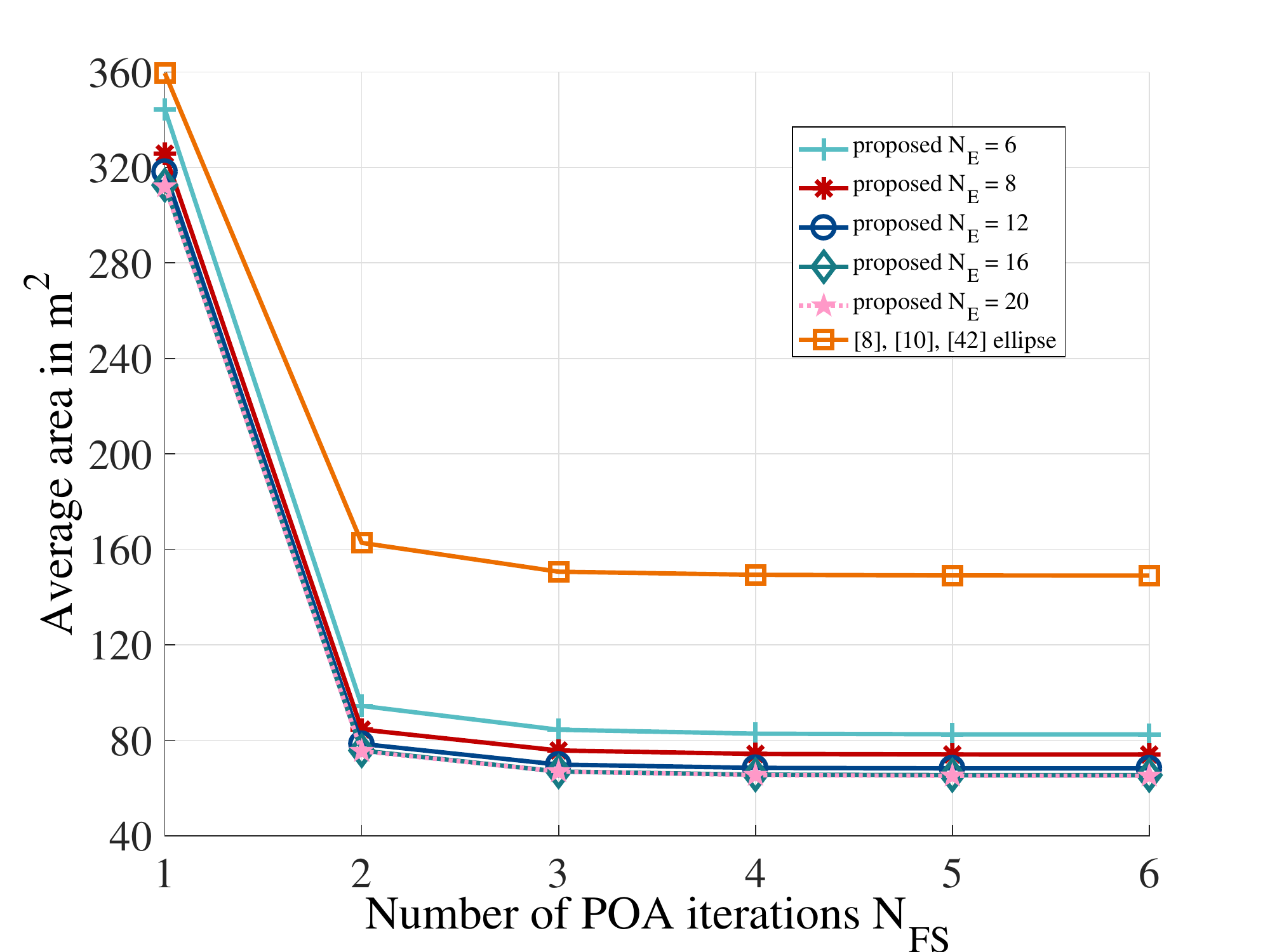}
	\caption{\textit{Average polygon/ellipse area versus the number of iterations} - It becomes evident that after a few iterations (2-3) the polygon/ellipse area does not reduce significantly. Moreover, a small number of polygon edges $N_E \approx 16$ is sufficient to achieve tight polygons. }
	\label{fig:polygon_area_convergence_multiple_Ns}
\end{figure}

In order to assess the complexity NBP localization, we conduct two analyses. First, we consider the number of operations required for each variant of the algorithm. Then, we measure the computation time which is required. All computations of the simulation are performed on an Intel i7-5820k desktop CPU which was exclusively dedicated to simulation. Three different aspects of the computation time are considered, namely the time required to determine the outer-approximating polygon, the time required to determine converged location estimates, and the accumulated time. Note that the time which is required to determine converged location estimates depends on the number of iterations required to achieve convergence\footnote{For us, convergence is achieved if the average localization error does not change notably for successive iterations.}. 

\subsection{Polygon Support Outer-Approximation}
\label{subsec:POA}
We first explore the influence of the number of polygon vertices on the area of the polygon. For comparability, we also analyze the ellipse outer-approximation algorithm which was presented in \cite{GWG2013,GWSR2011,GHO2011}. Fig. \ref{fig:polygon_area_convergence_multiple_Ns} depicts the polygon area in $\text{m}^2$ against the number of iterations. Recall that it is desirable to have polygons of small size in order to tightly constrain the positions of agents. We can obtain from Fig. \ref{fig:polygon_area_convergence_multiple_Ns} that the polygon size reduces, as the number of vertices increases. Increasing the number of vertices above $N_E>16$ adds no additional area reduction. Thus, we restrict ourselves to polygons with $N_E = 16$ vertices for all further analysis. In addition, we see that the largest area reduction is achieved in the $2^{\text{nd}}$ iteration. The reduction of polygon area from the $2^{\text{nd}}$ to the $3^{\text{rd}}$ iteration is only minor. Using more than $3$ iterations does not decrease the polygon area notably. To determine our proposal distributions as cheap as possible in terms of computation, we restrict ourselves to $N_{FS}=2$ iteration for all further analysis. 
\begin{figure}[t]
	\centering	      
	\includegraphics[width=.5\columnwidth]{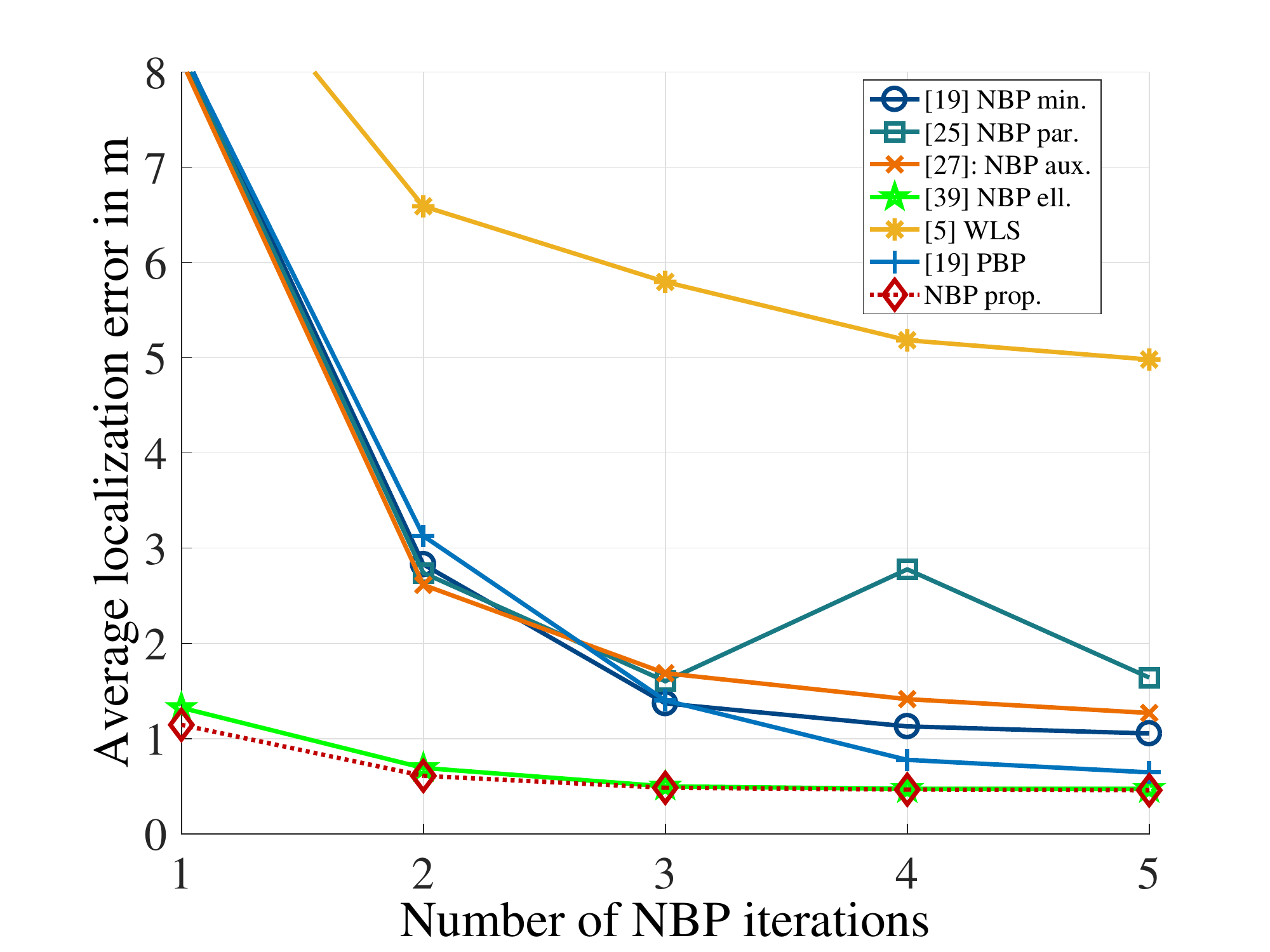}
	\caption{\textit{Average localization error} - The average localization error is depicted against the number of iterations. Our proposed NBP variant converges significantly faster.}
	\label{fig:average_localization_error_vs_niter}
\end{figure}
\subsection{Novel Proposal Distribution for Belief Propagation}
We now consider the impact of our novel proposal distribution on the localization accuracy when using nonparametric belief propagation. We consider the following reference sampling techniques: 1) \cite{IFMW2005} parsimonious sampling ('NBP par.'), 2) \cite{SZ2013} sampling based on an auxiliary variable ('NBP aux.'), 3) \cite{LFSWW2012} sampling based on the incoming message with the lowest entropy ('NBP min.'), and 4) \cite{MB2016} sampling with elliptical constraints on the sample space ('NBP ell.'). For better comparability, we also consider the weighted least squares ('WLS') approach from \cite{SLK2002} and the parametric belief propagation ('PBP') method from \cite{LFSWW2012}. We analyze the the speed of convergence, localization accuracy, the computation time, and the number of samples. 
\subsubsection{Convergence and Accuracy} 
Fig. \ref{fig:average_localization_error_vs_niter} depicts the average localization error against the number of iterations. The localization error decreases in every iteration until convergence is reached. From Fig. \ref{fig:average_localization_error_vs_niter}, we can infer two benefits of our proposal distribution. First, convergence is achieved quicker compared to the baseline approaches. Secondly, our approach achieves the highest localization accuracy among all considered algorithms. Hence there is the two-fold benefit of incorporating our proposal distribution. We can trace the previous two observations back to the following reason. Quick convergences is achieved since the polygonal constraints already restrict the possible locations, and the beliefs are concentrated in the areas close to the true locations from iteration one. High localization accuracy is achieved since samples are reside in the areas close to the true location, and, unlike in \cite{IFMW2005,SZ2013,LFSWW2012}, almost all particles contribute to the location estimate.

The gain in terms of localization accuracy becomes even more evident, when the outage probability after convergence\footnote{Most baseline approaches (\cite{IFMW2005} 'NBP par.', \cite{SZ2013} 'NBP aux.', \cite{LFSWW2012} 'NBP min.', and \cite{LFSWW2012} PBP) need 5 iterations to converge, while \cite{MB2016} 'NBP ell.' and \cite{SLK2002} 'WLS' need 2 and 10 iterations, respectively.} in Fig. \ref{fig:outage_probability} is considered. With our proposed NBP variant, the outage probability decreases rapidly in the regime of small errors, and especially, large errors can be mitigated better compared to the baseline approaches. Our proposal greatly outperforms all considered variants of NBP. Parsimonious sampling from \cite{IFMW2005}, shows poor performance in the regime of larger errors, which can be traced back to the fact that samples are drawn based on the belief of the previous iteration. When the previous belief was erroneous, the current belief will be impaired by the previous belief. This also explains the non-monotonic decrease of the average localization error in Fig. \ref{fig:average_localization_error_vs_niter}. The variant from \cite{SZ2013}, which is based on an auxiliary variable, considers only information from anchors to draw samples. Due to the sparseness of the anchors in the considered scenario, the samples are not concentrated tightly in the area of the true location. Thus, no considerable advantage can be seen in terms of accuracy, compared to the non-constrained sampling  approach from \cite{LFSWW2012}. The PBP variant provides the most accurate results among the benchmark schemes.
\begin{figure}[t]
	\centering	      
	\includegraphics[width=.5\columnwidth]{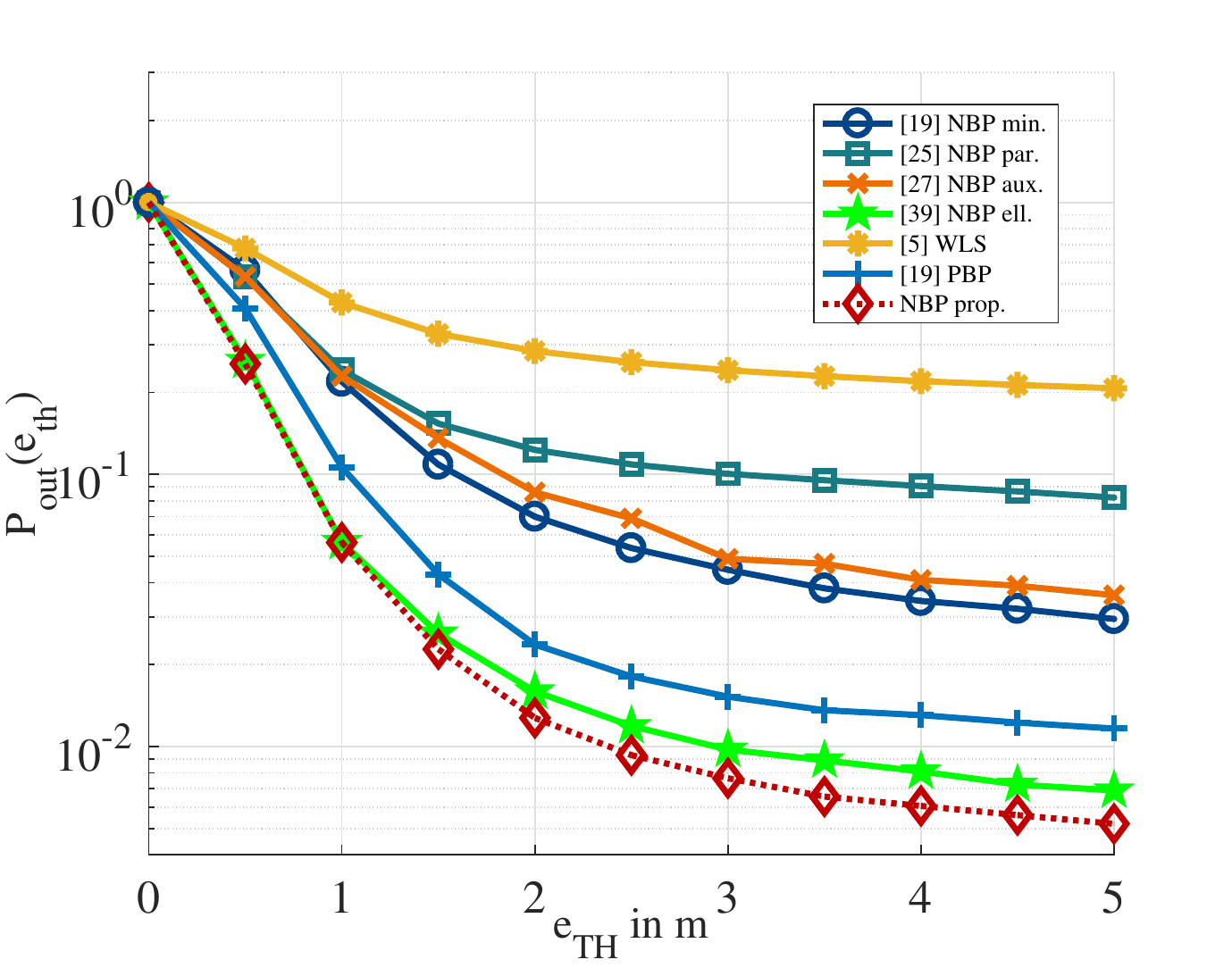}
	\caption{\textit{Outage probability compared to baseline algorithms} - Our proposed NBP variant outperforms all benchmarks in terms of localization accuracy.}
	\label{fig:outage_probability}
\end{figure}
\subsubsection{Computation Time}
Table \ref{tab:computation_time} depicts the average computation time $t_{\mathrm{c}}$ per agent. We break the accumulated computation time up into the time that is required to compute the polygons $t_{\mathrm{poly}}$ and the time required to achieve convergence with the respective localization algorithm $t_{\mathrm{conv}}$, i.e. $t_{\mathrm{c}}=t_{\mathrm{poly}}+t_{\mathrm{conv}}$. For the latter time, we use the convergence observations from the previous discussion.  
\begin{table*}%
\centering
\begin{tabular}{|c|c|c|c|c|c|c|c|}
\hline
\multirow{2}{*}{}& \multicolumn{7}{c|}{Algorithms} \\
& \cite{LFSWW2012} NBP  & \cite{SZ2013} NBP aux. & \cite{IFMW2005} NBP  par. & \cite{MB2016} NBP ell. &\cite{LFSWW2012} PBP & \cite{SLK2002} WLS & NBP prop. \\ \hline
\multicolumn{1}{|l|}{$t_{\mathrm{poly}}$ in s}& - & - 			& - 			& - 	&-	 & -	& 0.0019	\\ \hline
\multicolumn{1}{|l|}{$t_{\mathrm{conv}}$ in s}& 16.7189 	& 16.8082 & 19.0205 & 7.4425 & 0.6253 &  0.0618	& 7.3275 \\ \specialrule{.1em}{.05em}{.05em}

\multicolumn{1}{|l|}{$t_{\mathrm{c}}$ in s} &  16.7189	& 16.8082 & 19.0205 & 7.4425 & 0.6253 & 	0.0618	& 7.3294 \\ \hline
\end{tabular}
\begin{tabular}{c}
\end{tabular}
\caption{\textit{Computation time} - Average time per node to achieve convergence.}
\label{tab:computation_time}
\end{table*}
In terms of computation time, WLS shows the lowest cost followed by PBP and our proposed polygon-based NBP. We can infer two important conclusions
\begin{itemize}
	\item Computation time for polygon support outer-approximation is almost negligible 
	\item The increased speed of convergence reduces computation time significantly compared to the baseline NBP approaches
\end{itemize}
These observations meet our expectations on the computation time. To gain some more insight into the first observation, let us review polygon support outer-approximation. Note that all operations in that algorithm can be solved in closed form. Polygon scaling and anchor polygon processing (Algorithm \ref{alg:anchor_polygon} and \ref{alg:polygon_scaling}, respectively) both scale linear with the number of edges, i.e. $\mathcal{O}(N_E)$. In terms of computation, the most demanding part of the algorithm is polygon intersection (Algorithm \ref{alg:polygon_intersection}). The Sutherland-Hodgman algorithm scales quadratic in the number of edges, $N_E$ \cite{SH1974}. Since it intersects only pairs of polygons, the intersection of $|\mathcal{S}_{\rightarrow j}|$ polygons requires to execute the Sutherland-Hodgman algorithm $|\mathcal{S}_{\rightarrow j}|$ times. Recall that the number of edges of intersected polygons depends on the input polygons and it cannot be generalized. Thus, we cannot quantify the number of operations required to intersect $|\mathcal{S}_{\rightarrow j}|$ polygons in general. We observed that the number edges of the two intersecting polygons $\tilde{N}_E$ is typically less than the number of initial edges $N_E$. Hence $\tilde{N}_E \approx N_E$ over-estimates the number of computations. With this assumption, polygon intersection scales according to $\mathcal{O}((|\mathcal{S}_{\rightarrow j}|)N_E^2)$. Considering that $N_E\ll N_S$, it is obvious that NBP is much more costly in terms of computation than polygon outer-approximation. Observation 2) makes intuitively sense. If we consider 2), it becomes evident that reducing the number of iterations also reduces the complexity of NBP linearly. We see that our polygon-based NBP results in a considerable reduced computation time compared the baseline variants of NBP. In particular, a reduction of approximately $60\%$ is achieved. Yet, our proposed NBP variant has somewhat larger computational requirements (approximately a factor of $11$), when compared to the parametric approach. It should be noted, however, that the localization accuracy is generally higher with our proposal. This will become more evident in the following subsection. 

\subsubsection{Number of Samples}
Fig. \ref{fig:number_of_particles} depicts the outage probability after convergence considering different numbers of samples.  Here, we only consider PBP since it is the strongest competitor in terms of localization accuracy. Two observations can be made: 1) the number of samples should be sufficiently large in order to outperform PBP, and 2) when the number of samples grows, larger gains can be achieved compared the parametric approach. For large samples sizes ($N_{\mathrm{s}}\geq 1000$) significant gains can be achieved in the regimes of both small and large error. Yet, larger samples sizes result in higher computation times. Considering the results from the previous subsection, we can draw the following conclusion regarding the accuracy-computation trade-off. For systems which do not aim to maximize the localization accuracy, it may be sufficient to choose PBP for network localization, since it is generally cheaper in terms of computations. When more computational resources are available, however, NBP allows for a significant increase in terms of localization accuracy far beyond what is possible with PBP. Compared to other NBP variants, our polygon-based NBP provides a considerably improved accuracy-computation trade-off, i.e. higher localization accuracy is achieved at only a fraction of the computational costs compared to state-of-the-art variants of NBP. 
\begin{figure}[t]
	\centering	      
	\includegraphics[width=.5\columnwidth]{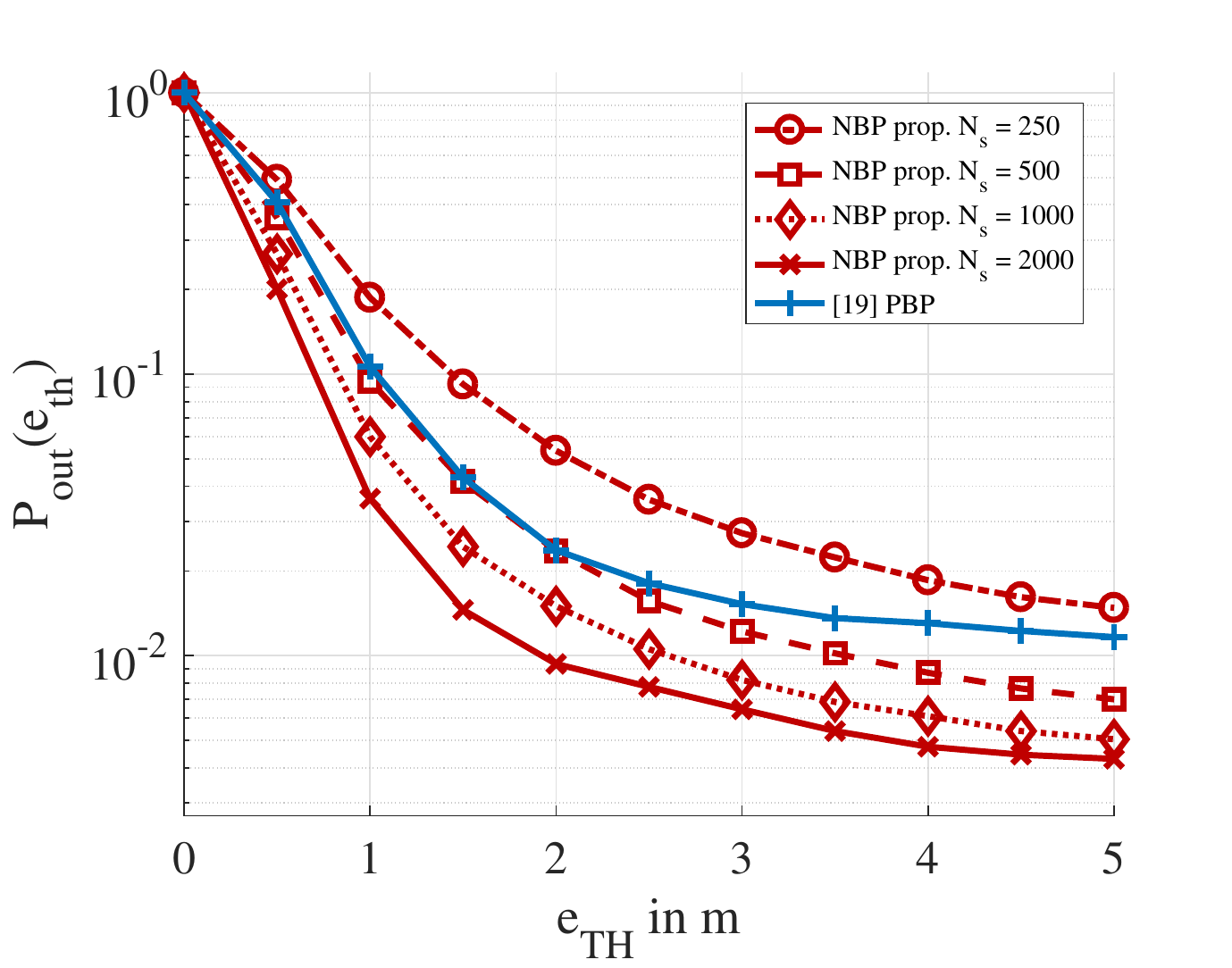}
	\caption{\textit{Outage probability for different sample sizes} - With our proposal distribution, the number of particles can be reduced by factor of four while still outperforming the proposal distribution from \cite{LFSWW2012} in terms of outage probability.}
	\label{fig:number_of_particles}
\end{figure}

\subsection{Discussion and Summary}
Generally, stochastic inference for cooperation localization is a hard task, when only a few anchor nodes are available. In networks with sparse anchor coverage, the positioning uncertainty is generally large. Nonparametric belief propagation is a powerful tool to perform inference as it approaches the performance of the maximum a posteriori estimator, when infinitely many samples are used. We saw in the previous section that there is a trade-off between localization accuracy and computational costs. Due to the computational costs, large sample sizes may become prohibitive, and small sample sizes have to be considered. We showed that state-of-the-art variants of NBP perform relatively poor for a small number of samples, because samples are not used efficiently, i.e. samples resides in areas of the sample space which are far away from the true location. When the anchor coverage is sparse, PBP can outperform NBP with small sample sizes in terms of accuracy \textit{and }computation time. To leverage the full potential of NBP at reasonable computational costs, we proposed a computationally cheap algorithm, called polygon outer-approximation, which confines the sample space to convex polygons. With these constraints, only a small number of samples is necessary to outperform both PBP and other NPB variants in terms of accuracy. 

\section{Conclusion}
\label{section_6}
We treat cooperative positioning in wireless networks as stochastic inference problems, and we proposed a polygon-constrained variant of nonparametric belief propagation to solve these problems. To relax the inference procedure, we split the problem into two stages. In the first stage, we determine constraints on the sample space, which confine the positions of the nodes to convex polygons. We also provide a mathematical proof under which conditions the sample space can be constrained. In the second stage, we solve a constrained stochastic inference problem using our polygon-constrained variant of nonparametric belief propagation to obtain estimates on the positions of nodes. Our proposal shows significantly increased localization accuracy and speed of convergence, compared to state-of-the-art cooperative positioning algorithms. At the same time, the computation time is reduced considerable compared to state-of-the-art nonparametric belief propagation variants. Hence polygon-constrained nonparametric belief propagation offers the benefit of highly accurate localization with reasonable computational costs. 

\appendix[Proof of Theorem \ref{theorem}]
Let us consider the factorized a posteriori distribution from (\ref{eq:factorized_APD}). Our goal is to confine the support of the marginal a posteriori distributions $p_{\vec{X}_j|\vec{Z}}(\vec{x}_j|\vec{z}), ~ \forall j$ to feasible sets. The marginal a posteriori distribution of $\vec{x}_j$ is formally defined as 
\begin{equation}
	p_{\vec{X}_j|\vec{Z}}(\vec{x}_j|\vec{z}) = \int p_{\vec{X}|\vec{Z}}(\vec{x}|\vec{z})~\text{d} {\raise.17ex\hbox{$\scriptstyle\sim$}}\left\{\vec{x}_j\right\}= \int{\prod_{k=1}^K {f_k(\vec{s}_k)}}~\text{d} {\raise.17ex\hbox{$\scriptstyle\sim$}}\left\{\vec{x}_j\right\},
	\label{eq:definition_marginal}
\end{equation}
where the joint a posteriori distribution $p_{\vec{X}|\vec{Z}}(\vec{x}|\vec{z})$ factorizes into $K$ factors $f_k$ which just depend on subsets of variables, i.e. $\vec{S_K} \subseteq \vec{X}$.
\begin{lemma}
	The support of the product of two densities $f_k(\vec{x}_j,\vec{s}_k\backslash\vec{x}_j) \cdot f_l(\vec{x}_j,\vec{s}_l\backslash\vec{x}_j)$ is determined by the intersection of the support of each density, i.e. $\textsc{supp}(f_k(\vec{x}_j,\vec{s}_k\backslash\vec{x}_j)f_l(\vec{x}_j,\vec{s}_l\backslash\vec{x}_j))=\textsc{supp}(f_k(\vec{x}_j,\vec{s}_k\backslash\vec{x}_j))\cap \textsc{supp}(f_l(\vec{x}_j,\vec{s}_l\backslash\vec{x}_j))$.
\label{lemma:domain_product}
\end{lemma}
\begin{proof}
	Consider a simpler case: $f(a,b)g(a,b)$. The support is determined by the closure of the subset of all $a \in A$ and $b\in B$ for which $f(a,b)g(a,b)\neq 0$, i.e. $\textsc{supp}(f(a,b)g(a,b))= \{a\in A, b\in B | f(a,b)g(a,b)\neq 0\}$. The product is either zero if one of the factors is zero or if both are. Thus, the contribution of $\textsc{supp}(f(a,b))$ to $\textsc{supp}(f(a,b)g(a,b))$ is $\{\textsc{supp}(f(a,b))\backslash a,b | g(a,b)=0 \}$. The same holds also for the contribution of $\textsc{supp}(g(a,b))$, i.e. $\{\textsc{supp}(g(a,b))\backslash a,b | f(a,b)=0 \}$. Thus, the support of the product of densities yields the intersection of the support of each density $\textsc{supp}(f(a,b)g(a,b))=\textsc{supp}(f(a,b))\cap\textsc{supp}(g(a,b))$. It is straightforward to extend this idea to more than two variables and factors.
\end{proof}
The support of the marginal a posteriori distribution $p_{\vec{X}_j|\vec{Z}}(\vec{x}_j|\vec{z})$ will be determined by the intersection of the support of the product of densities $\prod_{k=1}^K {f_k(\vec{s}_k)}$. The factors $f_k$ with argument $\vec{s}_k$ will not constrain the support of $\vec{x}_j$ if the sets $\vec{s}_k$ and $\vec{x}_j$ are disjoint, i.e. $\vec{x}_j \cap \vec{s}_k = \emptyset$.

\begin{lemma}
The support of $p_{\vec{X}_j|\vec{Z}}(\vec{x}_j|\vec{z})$ is independent of the domain of ${\raise.17ex\hbox{$\scriptstyle\sim$}}\left\{\vec{x}_j\right\}=\vec{x} \backslash \vec{x}_j$.
\label{lemma:independent_domain}
\end{lemma}

\begin{proof}
Consider the previous example in a slightly altered form: $h(a)=\int{f(a,b)g(a,b)\text{d}b}$. Also consider the support of the product of densities $\textsc{supp}(f(a,b)g(a,b))=\textsc{supp}(f(a,b))\cap\textsc{supp}(g(a,b))$. Integrating of the entire domain of $b$ does not constrain the support of $h(a)$, i.e. $\textsc{supp}(h(a))=\textsc{supp}(f(a)g(a))=\textsc{supp}(f(a))\cap\textsc{supp}(g(a))$. It is readily shown that this is also valid for multiple integration variables and more than two factors.
\end{proof}

Considering Lemma \ref{lemma:domain_product} and \ref{lemma:independent_domain}, we can infer that the support of the marginal a posteriori distribution $p_{\vec{X}_j|\vec{Z}}(\vec{x}_j|\vec{z})$ is compact, i.e. closed and bounded, if at least one factor $f_k(\vec{x}_j,\vec{s}_k\backslash \vec{x}_j)$ has compact support in $\vec{x}_j$. 

We now turn to the next vital condition which allows for computationally efficient outer-approximation, namely convexity of the support of the marginal a posteriori distribution $p_{\vec{X}_j|\vec{Z}}(\vec{x}_j|\vec{z})$. 

\begin{lemma}
The support of the marginal a posteriori distribution $p_{\vec{X}_j|\vec{Z}}(\vec{x}_j|\vec{z})$ is convex if and only if all factors $f_k(\vec{x}_j,\vec{s}_k\backslash \vec{x}_j)$ have convex support in $\vec{x}_j$. 
\label{lemma:convexity}
\end{lemma}

\begin{proof}
The support of the marginal a posteriori distribution is given by the intersection of the support of each factor $f_k(\vec{x}_j,\vec{s}_k\backslash \vec{x}_j)$ (Lemma \ref{lemma:domain_product}). The intersection of sets is convex if and only if all sets are convex \cite[p. 36]{BV2004}. Hence the intersection of all supports is be convex if and only if the support of each density $f_k(\vec{x}_j,\vec{s}_k\backslash \vec{x}_j)$ is convex. 
\end{proof}


%


%
%
%
%
%

\ifCLASSOPTIONcaptionsoff
  \newpage
\fi



%
\bibliographystyle{IEEEtran}
\bibliography{bibliography}

\end{document}